\documentclass[11pt]{article}
\usepackage{amsmath}
\usepackage{amsthm}
\usepackage{amsfonts}
\usepackage{amssymb}
\usepackage{latexsym} 
\usepackage{graphicx}
\usepackage{placeins}
\usepackage{cite}

\usepackage[matrix,tips,graph,curve]{xy}

\usepackage{float}
\restylefloat{figure}
\usepackage[font = footnotesize]{caption}
\usepackage{subcaption}



\makeatletter 
\@addtoreset{equation}{section}
\makeatother

\theoremstyle{plain}
\newtheorem{theorem}[equation]{Theorem}
\newtheorem{corollary}[equation]{Corollary}
\newtheorem{lemma}[equation]{Lemma}
\newtheorem{proposition}[equation]{Proposition}

\theoremstyle{definition}
\newtheorem{definition}[equation]{Definition}

\newcommand\reals{{\mathbb R}}

\def\km{k\text{-}m}
\def\Lk{\mathrm{Lk}}
\def\Pr{\mathrm{Pr}}
\def\cone{\text{cone}}
\def\len{\operatorname{length}}
\def\Int{\text{Int}_{n}}
\def\Intkm{\text{Int}_{k+m}}
\newcommand{\od}{\stackrel{\mbox {\tiny {def}}}{=}}
\def\CF{\mathrm{CF}}
\def\min{\mathrm{min}}
\def\supp{\operatorname{supp}}
\def\mod{\operatorname{mod}}

\begin{document}

\title{Periodic Neural Codes and Sound Localization in Barn Owls}
\author{Lindsey S. Brown \and Carina Curto}
\date{July 21, 2021}
\maketitle

\begin{abstract}
Inspired by the sound localization system of the barn owl, we define a new class of neural codes, called periodic codes, and study their basic properties. Periodic codes are binary codes with a special patterned form that reflects the periodicity of the stimulus. Because these codes can be used by the owl to localize sounds within a convex set of angles, we investigate whether they are examples of convex codes, which have previously been studied for hippocampal place cells. We find that periodic codes are typically {\it not} convex, but can be completed to convex codes in the presence of noise. We introduce the convex closure and Hamming distance completion as ways of adding codewords to make a code convex, and describe the convex closure of a periodic code. We also find that the probability of the convex closure arising stochastically is greater for sparser codes. Finally, we provide an algebraic method using the neural ideal to detect if a code is periodic. We find that properties of periodic codes help to explain several aspects of the behavior observed in the sound localization system of the barn owl, including common errors in localizing pure tones.
\end{abstract}

\section{Introduction}
Neural codes are patterns of neural activity, also known as \textit{codewords}, that arise from the encoding of environmental stimuli.  Understanding neural codes means understanding both their structure and the relationship between these codewords and the stimuli they represent.  One way to begin understanding this structure is by considering the neural code as a neural ring and exploring the intrinsic combinatorial properties as they relate to the structure of the stimulus space in an algebraic framework \cite{CurtoRing, CurtoBook}.  More recent work has focused specifically on the relationship between the code and the stimulus space by considering whether each neuron fires over a convex region of space, motivated by the place cells in the rat hippocampus \cite{NewPaper, CurtoConvex, CurtoRev}.

Inspired by the owl auditory system, in this work, we focus on {\it periodic codes}, which have a special structure that may be especially well-suited for encoding stimuli that are similarly periodic, such as sound waves.  Periodic firing patterns are observed in the nucleus laminaris of the barn owl, the first site of binaural convergence in the auditory pathway.  Similarly, the codewords in a periodic code, $C_{k,m}(n)$, have a precise, periodic pattern: they consist of bands of $k$ consecutive neurons that are firing, alternating with bands of $m$ consecutive neurons that are silent (see Figure~\ref{CodeExample}).  One advantage of binaural hearing is the ability to localize sounds, and we explore how the structure of periodic codes relates to this ability to localize a sound to a convex set of angles by asking whether periodic codes are convex.  

First, we show that, except in trivial cases, periodic codes are not convex in Theorem~\ref{subcompletionTheorem}, which may explain the barn owl's errors in sound localization when presented with a single frequency stimulus.  This theorem also gives a specific formulation for adding codewords to make these codes convex, and we define the \textit{convex closure} of a code.  In the case of periodic codes, the convex closure involves taking a union with another periodic code, suggesting that owls may resolve the issues of nonconvexity of single frequency codes by combining codes from multiple frequencies higher in the brainstem.   Second, we give the precise probability that the convex closure arises instead from stochasticity, finding that sparser codes are more likely to be completed to convex codes via stochastic processes.  Third, we give an algorithmic method to determine if an arbitrary code, with the neurons labelled in a potentially permuted order, is periodic using the neural ring (Theorem~\ref{fullCF}).  

The organization of this paper is as follows. In Section 2, we give a rigorous definition of periodic codes and explore basic combinatorial properties of these codes, highlighting the differences between periodic codes and cyclic codes \cite{CyclicCodeBook}.  In Section 3, we prove our main result on the convexity properties of these periodic codes,  Theorem~\ref{subcompletionTheorem}.  In Section 4, we explore the role of stochasticity in creating convex codes from non-convex periodic codes, deriving the probability that this transformation occurs. In Section 5, we conclude by analyzing periodic codes from an algebraic perspective and prove Theorem~\ref{fullCF}.

\section{Periodic codes}
We define a special class of codes, \textit{periodic codes}, which have both combinatorial and biological significance.  In this section, we give the basic properties of these codes and compare them to the more familiar \textit{cyclic codes}.  We end with a description of sound localization in the barn owl, which is our motivating
biological example of periodic codes, and discuss questions that arise from considering periodic codes in this context.

\subsection{Basic definitions}

We first introduce combinatorial neural codes.  To relate continuous firing patterns of a set of $n$ neurons to a discrete object, we associate each of the $n$ bits in a binary string to a neuron $x_i$, where $x_i = 1$ if the firing rate $f(x_i) \geq t$ for some firing threshold $t$ and $x_i = 0$ otherwise.  A \textit{neural code, C} of length $n$ is the collection of these binary strings, where each binary string is a \textit{codeword, c} of $C$.  Note that we will interchangeably use $c$ to denote a binary string and the set of indices in $[n]$ which are 1 in the binary string representation.  For example, $c = 10100$ is equivalent to $\{1, 3\}$.  This discrete formulation allows us to explore combinatorial and topological properties of neural firing.  To relate the codewords to the encoded stimuli, by analogy with place field codes, we are able to define subsets of the stimulus space for which $x_i =1$ as the receptive field of neuron $i$.  

An abstract {\it simplicial complex} $\Delta$ is a collection of sets which is closed under the operation of taking subsets, meaning if $\sigma \in \Delta$ and $\tau \subset \sigma$, then $\tau \in \Delta$.  Each element of a simplicial complex is called a \textit{face}, and a face $\sigma$ has \textit{dimension} $|\sigma| - 1$.  If a face is maximal, in the sense that it is not a proper subset of any other element of $\Delta$, then it is called a {\it facet}. Every neural code $C$ has a corresponding simplicial complex $\Delta(C)$ \cite{CurtoConvex}, defined as follows. 

\begin{definition}
The simplicial complex of a code $C$, denoted $\Delta(C),$ is given by
$$\Delta(C) = \{\sigma \mid \sigma \subset c \text{ for some } c \in C\}.$$
\end{definition}
Each neuron is viewed as a vertex, and the subsets of neurons which cofire in each of the different codewords correspond to higher dimensional simplices.  The simplicial complex of a code provides a useful topological structure but loses much of the detailed information about the code \cite{CurtoRing}.

\subsection{Periodic codes and their properties}
We now formally introduce periodic codes.  For any $k, m \in \mathbb{N}$, let $s_{k,m}$ denote the binary string $s_1\cdots s_{k+m}$ such that $s_i=1$ for $1\leq i \leq k$ and $s_j = 0$ for $k+1 \leq j \leq k+m$.  For example, $s_{2,3} = 11000$.

We use the term \textit{substring} to refer to a subset of bits of consecutive indices, where an \textit{x-substring} is a substring of length $x$.  For example, 010 is a 3-substring of 10101, but 111 is not because the indices are not consecutive.

\begin{definition}
{Let $k$, $m$, and $n$ be nonnegative integers such that $n\geq k+m$.  Let $c = c_1c_2\cdots c_n$ be a codeword of length $n$.  We say $c$ is a \textit{$\km$ periodic codeword on n neurons} if every $(k+m)$-substring of $c$ is a cyclic permutation of $s_{k,m}$.}
\label{codeword}
\end{definition}

More informally, a $\km$ periodic codeword on $n$ neurons is a codeword of length $n$, consisting of bands of activity and inactivity, where all of the bands of consecutive 1's have length $k$ and all the bands of consecutive 0's have length $m$ with the possible exceptions of the first and last band of activity or inactivity which may be  shorter.  This is illustrated in Figure \ref{CodeExample}, where the periodic pattern consists of bands of $k=2$ active neurons followed by bands of $m=3$ inactive neurons, but in the first band of activity only a single neuron fires.

\begin{figure}
	\centering
		\includegraphics[width=.9\textwidth]{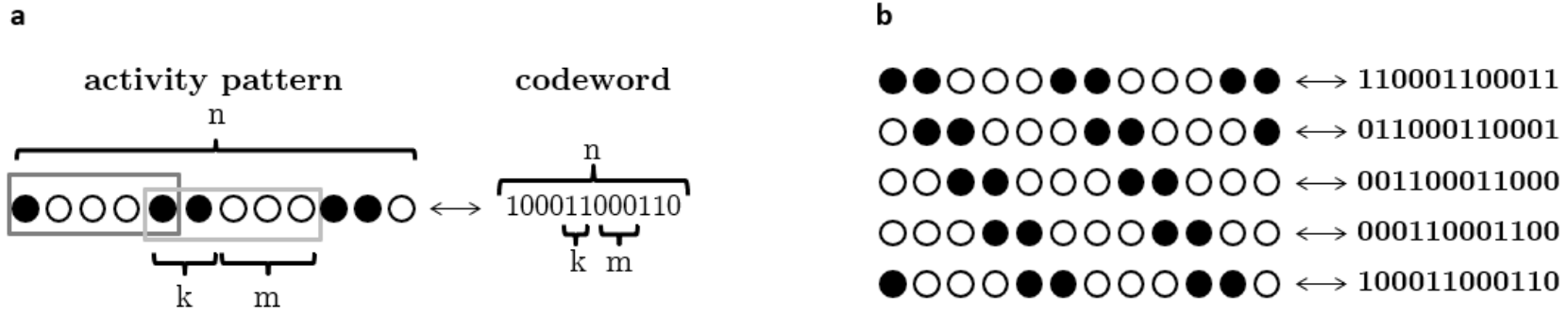}
	\caption{\emph{k-m Periodic Codes.}  \textbf{(a)} A $\km$ periodic codeword on $n$ neurons, where $k=2$, $m=3$, and $n=12$.  Circles depict active (filled) and inactive (no fill) neurons.  Each box shows a $(k+m)$-substring corresponding to cyclic permutations 10001 and 11000 of the fundamental string $s_{2,3} = 11000$.  \textbf{(b)} The 2-3 periodic code on 12 neurons, $C_{2,3}(12)$.  Each codeword begins with a different permutation of $s_{2,3}$.  Note that $|C_{2,3}(12)| = k+m = 5$.}
	\label{CodeExample}
\end{figure}

Recall that the \textit{Hamming weight} of a binary string $b$ is given by $w_H(b) = \sum b_i$.   

\begin{lemma}
\label{weightProp}
{\emph{Uniform Weight Property.} Let $c = c_1c_2\cdots c_n$ be a $\km$ periodic codeword on $n$ neurons.  Every substring of length $k+m$ has Hamming weight $k$.}
\end{lemma}

Lemma \ref{weightProp} follows immediately from Definition \ref{codeword} since every $(k+m)$-substring of $c$ is a cyclic permutation of $s_{k,m}$, which has weight $k$.  Observe that, although each codeword has the uniform weight property, every codeword does not have the same weight; in Figure \ref{CodeExample}b, the first codeword has weight 6, the second has weight 5, and the third has weight 4.

As the name suggests, $\km$ periodic codewords exhibit a periodic property, where the first $k+m$ bits of the codeword repeat periodically as formalized in Lemma \ref{periodicityProp} below. 

\begin{lemma}
\label{periodicityProp}
{\emph{Periodicity Property.} Let $c = c_1c_2\cdots c_n$ be a $\km$ periodic codeword on n neurons.  If $i \equiv j\mod{(k+m)}$, then $c_i = c_j$.}
\end{lemma}

\begin{proof}
We will show that if $j=i+(k+m)$, then $c_i = c_j$.  From here, it follows by transitivity that $c_i = c_j$ for all $i,j$ such that $i \equiv j \mod{(k+m)}$.  Recall that $n\geq k+m$.  If $n=k+m$, there is nothing to prove.  Suppose $n>k+m$.  Let $i,j \in [n]$ with $j=i+(k+m)$.  Consider the $(k+m)$-substrings $s_1 =c_i\cdots c_{i+(k+m)-1}$ and $s_2 = c_{i+1}\cdots c_j$.  By Lemma \ref{weightProp}, $w_H(s_1)=w_H(s_2)=k$. Since $s_1$ and $s_2$ overlap on $k+m-1$ bits, we have that $w_H(s_2) = w_H(s_1)-c_i+c_j$.  Thus, we can conclude $c_i = c_j$.
\end{proof}

These two properties yield additional characterizations of $\km$ periodic codewords.

\begin{lemma}
\label{equivChar}
Let $c$ be a binary codeword of length $n$.  The following are equivalent:
\begin{enumerate}
\item{$c$ is a $\km$ periodic codeword on $n$ neurons.}
\item{Every $(k+m)$-substring of $c$ is a cyclic permutation of $s_{k,m}$.}
\item{$c_1 \cdots c_{k+m}$ is a cyclic permutation of $s_{k,m}$, and every $(k+m)$-substring has weight $k$.}
\item{$c_1 \cdots c_{k+m}$ is a cyclic permutation of $s_{k,m}$, and for all $i,j \in [n]$, if $i \equiv j \mod{(k+m)}$, then $c_i = c_j$.}
\end{enumerate}
\end{lemma}

\begin{proof}
($1\Leftrightarrow2$) The equivalence between 1 and 2 follows is given by Definition \ref{codeword}.  ($2\Rightarrow3$)  The equivalence from 2 to 3 follows directly from Lemma \ref{weightProp}.  ($3\Rightarrow4$) The proof is the same as that of Lemma \ref{periodicityProp}. ($4\Rightarrow2$) Suppose $c_i \cdots c_{i+k+m-1}$ is a cyclic permutation of $s_{k,m}$ and $c_i = c_{i+k+m}$.  Then, $c_{i+1} \cdots c_{i+k+m}$ is a cyclic permutation of $s_{k,m}$.  By hypothesis, $c_1 \cdots c_{k+m}$ is a cyclic permutation of $s_{k,m}$ and $c_i = c_j$ for all $i,j \in [n]$ such that $i \equiv j \mod{(k+m)}$.  So, by induction, every $(k+m)$-substring is a cyclic permutation of $s_{k,m}$.  
\end{proof}

These characterizations show that once the first $k+m$ bits of a $\km$ periodic codeword $c$ are given, all other bits of $c$ are determined.  This observation makes it easy to count the number of possible $\km$ periodic codewords and shows that the number of possible codewords is independent of $n$.

\begin{definition}
{The \textit{$\km$ periodic code on n neurons}, denoted $C_{k,m}(n)$, is the binary code which contains all possible $\km$ periodic codewords on $n$ neurons and no other codewords.} \end{definition}

Recall that the \textit{size} of a code $C$, denoted $|C|$, is the number of codewords it contains. 

\begin{proposition}
\label{size}
Let $C_{k,m}(n)$ be the $\km$ periodic code on $n$ neurons.
\begin{enumerate}
\item{If $k=0$ or $m=0$, then $|C_{k,m}(n)| = 1$.}
\item{If $k,m\neq 0$, then $|C_{k,m}(n)| = k+m$.}
\end{enumerate}
\end{proposition}
\begin{proof}
1.  If $k=0$ or $m=0$, then $s_{k,m}$ is a string of all 0's or all 1's.  Since there is only one possible permutation of $s_{k,m}$, we have $|C_{k,m}(n)| = 1$.  

2.  If $k,m \neq 0$, there are $k+m$ cyclic permutations of $s_{k,m}$.  There is one codeword beginning with each of these permutations, so $|C_{k,m}(n)| = k+m$.
\end{proof}

Figure \ref{CodeExample} shows an example of the firing patterns and corresponding codewords of a periodic code.  The $\km$ periodic code is completely parameterized by $k$, $m$, and $n$, which are 2, 3, and 12 in the figure respectively.  These five codewords constitute $C_{2,3}(12)$.

The previous properties resulted from the combinatorial properties of periodic codes, but the periodic structure of these codes also gives rise to topological properties in the simplicial complex of the code, which allow us to compare periodic codes to cyclic codes in the next section.    As a result of the Periodicity Property (Lemma \ref{periodicityProp}) of periodic codewords, we are also able to give a property of $\Delta(C_{k,m}(n))$, which will be useful for proving later results.

\begin{proposition}
Let $\Delta = \Delta(C_{k,m}(n))$.  Assume $i \equiv j \mod{(k+m)}$.  If $v_i \cup \sigma \in \Delta$, then $v_i \cup v_j \cup \sigma \in \Delta$.  In particular, $v_i \cup \sigma \in \Delta$ if and only if $v_j \cup \sigma \in \Delta$.
\label{simplexStructure}
\end{proposition}
\begin{proof}
Without loss of generality, assume $v_i \cup \sigma \in \Delta$.  Since by Lemma \ref{periodicityProp}, $c_i = c_j$ for all $c \in C$, $v_j$ is connected to $v_i$ and $v_j$ is connected to any face to which $v_i$ is connected.  Therefore, $v_i \cup v_j \cup \sigma \in \Delta$.  Since $\Delta$ is a simplicial complex, $v_j \cup \sigma \in \Delta$.
\end{proof}

\subsection{Comparison to cyclic codes}
One class of highly structured codes that are of particular relevance to coding theorists are \textit{cyclic codes} \cite{CyclicCodeBook}.  A cyclic code is defined by the property that the set of codewords is closed under all shifts in coordinates.  Since periodic codes have a similar periodic property (Lemma \ref{periodicityProp}), it is natural to ask whether periodic codes are just a special case of cyclic codes.

\begin{definition}
A \textit{cyclic code of length n} is a code $C$ with the property that for every $c_1c_2 \cdots c_{n-1}c_n \in C$, the cyclic permutation $c_nc_1c_2\cdots c_{n-1} \in C$.
\label{cyclicCode}
\end{definition}

Note that cyclic codes are often defined with the additional property that the code be \textit{linear}.  A \textit{linear binary code} $C$ is a binary code where for all $c, d \in C$, $c+d \in C$, where addition is performed bitwise over $\mathbb{F}_2$.  Clearly, the all-zeros codeword $0 0 \cdots 0 \in C$ for any linear code $C$.  We do not require the extra structure imposed by linearity because most periodic codes are not linear, and in fact, a $\km$ periodic code is linear if and only if $k=0$ since these are the only $\km$ periodic codes which contain the all-zeros codeword.

The following lemma shows that, although they are not linear, many periodic codes satisfy the cyclic property from Definition \ref{cyclicCode}.  

\begin{lemma}
Let $C = C_{k,m}(n)$.  The code $C$ is cyclic if and only if for all $c \in C$, $c_n = c_{k+m}$.
\label{cycEquiv}
\end{lemma}

\begin{proof}
Recall $n \geq k+m$.  ($\Rightarrow$)  Suppose $C$ is cyclic, and let $c_1\cdots c_n \in C$.  Since $C$ is cyclic, $c_nc_1 \cdots c_{n-1} \in C$.  By Lemma \ref{weightProp}, $w_H(c_1 \cdots c_{k+m}) = w_H(c_nc_1 \cdots c_{k+m-1}) = k$.  This implies $c_n = c_{k+m}$.  ($\Leftarrow$) Let $c = c_1 \cdots c_n \in C$.  To show $C$ is cyclic, we want to show $c' = c_nc_1 \cdots c_{n-1} \in C$.  By assumption $c_n = c_{k+m}$, so $c_nc_1 \cdots c_{k+m-1}$ is a cyclic permutation of $s_{k,m}$ and has weight $k$.  All other $(k+m)$-substrings of  $c'$ are substrings of $c$, and so they have weight $k$.  By part 3 of Lemma \ref{equivChar}, $c'$ is $\km$ periodic and hence in $C$.
\end{proof}

To see why $n$ must be a multiple of $k+m$, consider $C = C_{2,2}(5)$.  Clearly, the codeword $11001\in C$, but the cyclic permutation $11100 \notin C$.

\begin{proposition}
\label{periodicCyclic} 
Let $C_{k,m}(n)$ be a $\km$ periodic code on $n$ neurons.
\begin{enumerate}
\item{If $k=0$ or $m=0$, then $C_{k,m}(n)$ is a cyclic code, independent of $n$.}
\item{If $k\neq 0$ and $m\neq 0$, then $C_{k,m}(n)$ is a cyclic code if and only if $n$ is a multiple of $k+m$.}
\end{enumerate}
\end{proposition}

\begin{proof}
1. If $k=0$ or $m=0$, $C_{k,m}(n)$ consists only of the all 0's or all 1's codeword respectively, so is trivially cyclic.  2. Suppose $k,m \neq 0$.  Let $C = C_{k,m}(n)$ and recall $n\geq k+m$.  By Lemma \ref{cycEquiv}, it suffices to prove that $c_n = c_{k+m}$ for all $c \in C$ if and only if $n$ is a multiple of $k+m$.  ($\Leftarrow$)  Assume $n$ is a multiple of $k+m$.  By Lemma \ref{periodicityProp}, $c_n = c_{k+m}$.  ($\Rightarrow$)  We prove the contrapositive.  Assume $n = a(k+m) + b$ for integers $a,b>0$ and $b<k+m$.  By Lemma \ref{periodicityProp}, $c_n = c_b$ for all $c\in C$.  Since there exists a cyclic permutation of $s_{k,m}$, and hence a codeword $c \in C$, such that $c_b \neq c_{k+m}$, it follows that $c_n \neq c_{k+m}$ for that codeword.
\end{proof}

\begin{figure}
	\centering
		\includegraphics[width=.9\textwidth]{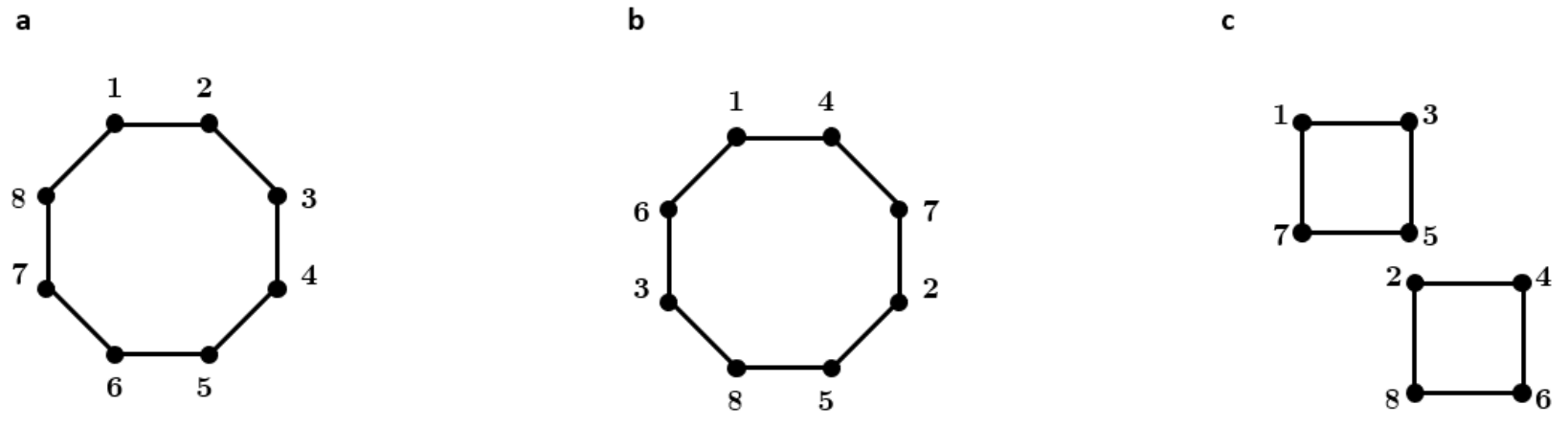}
	
	\caption{\emph{Comparing $\km$ Periodic Codes and Cyclic Codes.}  The simplicial complexes of \textbf{(a)} $C = C_{2,6}(8)$, \textbf{(b)} the cyclic code $C' =$ \{10010000, 01001000, 00100100, 00010010, 00001001, 10000100, 01000010, 00100001\}, and \textbf{(c)} the cyclic code $\tilde{C} =$ \{10100000, 01010000, 00101000, 00010100, 00001010, 00000101, 10000010, 01000001\}.  The codes $C$ and $C'$ have isomorphic simplicial complexes and so are permutation equivalent (Corollary \ref{maxCyclic}) as can be seen by matching vertices in the same position to obtain the permutation (24)(37)(68).  $\tilde{C}$ has a different simplicial complex and, thus, is not a $\km$ periodic code.}
	\label{CyclicPerm}
\end{figure}

Proposition \ref{periodicCyclic} tells us which $\km$ periodic codes are cyclic, so it is natural to ask which cyclic codes are permutation equivalent to $\km$ periodic codes, meaning that there exists a permutation of the vertices of the cyclic code such that the permuted code is periodic.  As seen in Figure~\ref{CyclicPerm}, some cyclic codes can be made periodic by applying a permutation of the vertices.  Comparing the simplicial complexes of a cyclic and periodic code allows us to see they are permutation equivalent when they have the same simplicial complex.  By matching the vertices in the simplicial complexes, we are able to give a permutation which makes the codes the same.  In this example, we can apply the permutation (24)(37)(68) to $C_{2,6}(8)$ to obtain $C'$.

We say that a codeword is \textit{maximal} if it is contained in no other codewords.  We call a code maximal if it contains only maximal codewords.  Observe that $C_{k,m}(n)$ is always maximal.

\begin{proposition}
Any two maximal codes are permutation equivalent if and only if they have isomorphic simplicial complexes.
\label{simpPermEquivalent}
\end{proposition}

\begin{proof}
($\Rightarrow$) It is clear that if two codes do not have the same simplicial complex, then they are not permutation equivalent.  ($\Leftarrow$) Let $C_1$ and $C_2$ be two maximal codes with isomorphic simplicial complexes.  Since all the codewords in $C_1$ and $C_2$ are maximal, they all correspond to a facet of the simplicial complex.  The isomorphism between $\Delta(C_1)$ and $\Delta(C_2)$ is a permutation of vertices that takes facets to facets and thus induces an isomorphism between $C_1$ and $C_2$ by permuting vertices (neurons).
\end{proof}

Observe that two codes with the same simplicial complex need not be permutation equivalent.  This is because two codes have the same simplicial complex if and only if they have the same maximal codewords, but the codes may differ on non-maximal codewords.  For example, consider the codes $C_1 = \{11\}$ and $C_2 = \{11, 10\}$.  We have $\Delta(C_1) = \Delta(C_2) = \{\emptyset, \{1\}, \{2\}, \{1,2\}\}$.  However, it is clear that the codes are not equivalent as $|C_1| \neq |C_2|$.  Thus, the maximality property is necessary.\\
\indent As a consequence of Proposition \ref{simpPermEquivalent}, we can now assert when a cyclic code is permutation equivalent to a periodic code.

\begin{corollary}
A cyclic code $C$ of length $n$ is permutation equivalent to $C_{k,m}(n)$ if and only if $\Delta(C) \cong \Delta(C_{k,m}(n))$ and $|C| = |C_{k,m}(n)|$.
\label{maxCyclic}
\end{corollary}

\begin{proof}
($\Rightarrow$) Assume $C$ is permutation equivalent to $C_{k,m}(n)$.  Clearly, $|C| = |C_{k,m}(n)|$.  Since $C_{k,m}(n)$ is maximal, $C$ must also be maximal.  By Proposition \ref{simpPermEquivalent}, $\Delta(C) \cong \Delta(C_{k,m}(n))$.  ($\Leftarrow$) Assume $|C| = |C_{k,m}(n)|$ and $\Delta(C) \cong \Delta(C_{k,m}(n))$.  Since $C_{k,m}(n)$ is maximal, its codewords are all the facets of $\Delta(C_{k,m}(n))$.  Since $\Delta(C) \cong \Delta(C_{k,m}(n))$, $C$ must also contain a codeword corresponding to each facet.  Since $|C| = |C_{k,m}(n)|$, $C$ must also be maximal.  By Proposition \ref{simpPermEquivalent}, $C$ and $C_{k,m}(n)$ are permutation equivalent.
\end{proof}

\subsection{Biological motivation: sound localization in the barn owl}

\begin{figure}
	\centering
		\includegraphics[width=\textwidth]{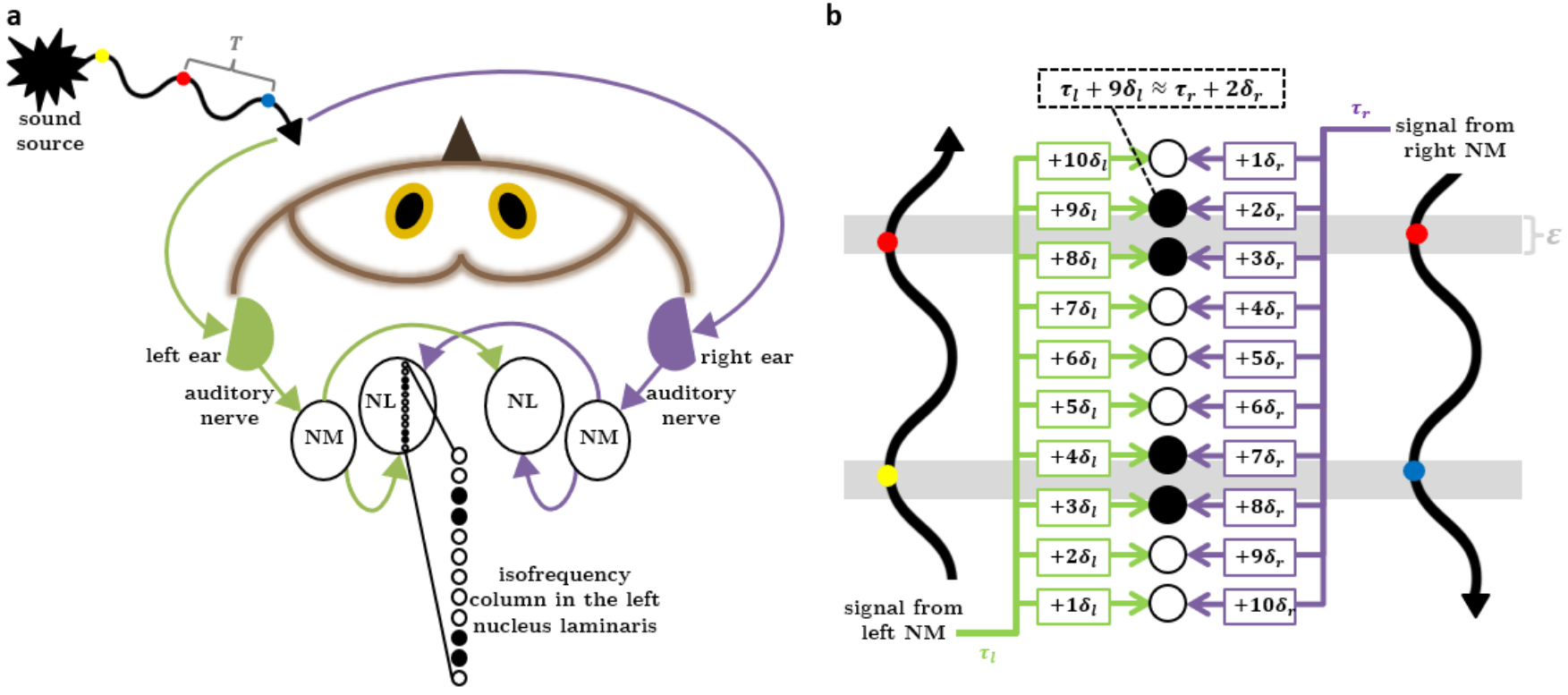}
	
	\caption{\textbf{(a)}\emph{Overview of the Barn Owl's Auditory Pathway.}  Sound waves travel to each of the owl's ears, stimulating the auditory nerve, which sends a signal to the nucleus magnocellularis (NM).  The NM then projects tonotopically to an isofrequency column in the ipsilateral (same side) nucleus laminaris (NL), entering on the dorsal side, and to an isofrequency column in the contralateral (opposite side) NL, entering on the ventral side.  \textbf{(b)}\emph{ Isofrequency Column of the Left Nucleus Laminaris.}  A sound wave, with period $T$ from a sound source closer to the left ear arrives at an isofrequency column of the left NL at a time delay of $\tau_{\ell}$ from the left (green) and at $\tau_r$ from the right (purple), with $\tau_l < \tau_r$.  Within the NL, there are different delays due to the depth the signal has traveled into the nucleus, which vary linearly with depth but at a different rate for the signals coming from each side of the brain, shown by changes by a factor of $\delta_l$ on the left and $\delta_r$ on the right.  A neuron in the column fires whenever it receives stimulation that is in phase from both sides within some error bound, $|(\tau_{\ell}+a\delta_{\ell})-(\tau_r+b\delta_r)|\mod{T} < \varepsilon$.  For example, the second neuron in the column fires if $|(\tau_{\ell}+9\delta_{\ell})-(\tau_r+2\delta_r)|\mod{T} < \varepsilon$.  Due to the difference in delay, a peak stimulating the left ear must travel deeper into the column than the same peak stimulating the right ear for the two peaks to coincide (red).  When these peaks coincide, other peaks will also coincide, an earlier point in the sound wave (blue) on the right coincides with a later point in the sound wave (yellow) on the left.}
	\label{AuditorySystem}
\end{figure}

From the comparison to cyclic codes, we see that the study of periodic codes is interesting because they share properties with some cyclic codes, but these codes are also interesting biologically because they give an abstraction of the neural firing in the barn owl's auditory system.  Barn owls use two cues to localize sounds in space, interaural intensity differences to determine the elevation of the sound source and interaural time differences to determine its azimuth.  Here, we give an overview of the interaural time difference pathway (Figure~\ref{AuditorySystem}a) and show how this pathway results in periodic firing in one of the nuclei (Figure~\ref{AuditorySystem}b). 

As shown in Figure \ref{AuditorySystem}a, the interaural time difference pathway begins in the nucleus magnocellularis (NM), which responds in a phase locked fashion to the incoming sound waves.  The NM projects onto the nucleus laminaris (NL), the first place of binaural convergence in the time difference pathway (signals from the left in green, signals from the right in purple).  As a result of the tonotopic projections from the NM, neurons in the NL are arranged tonotopically in isofrequency laminae, meaning neurons within a column fire only in response to a certain sound frequency.  It is within each of these isofrequency columns that we see periodic codes arise.

Figure \ref{AuditorySystem}b illustrates an isofrequency column of the NL. The ipsilateral signal, the signal coming from NM of the same side, enters through the dorsal surface and the contralateral signal, the signal coming from the NM of the opposite side, enters through the ventral surface of the NL.  A neuron in this column acts as a coincidence detector, firing when it receives simultaneous stimulation from both sides, analogous to the model of delay lines proposed by Jeffress for the mammalian medial superior olive.  To show periodic firing in the column, we first compute the delay to each neuron in the signals from each side of the brain.  

A sound source on the horizon travels a different distance to reach each ear and this signal must be transmitted through the auditory pathway on each side of the brain before reaching a column of the NL, giving us different time delays from each side, $\tau_{\ell}$ and $\tau_r$, before the signal reaches the NL.  Once these signals enter the NL, experiments show that the conduction delay varies linearly with depth; the ipsilateral side changes at approximately .46 degrees per micrometer ($\delta_l$), and the contralateral side changes at approximately .68 degrees per micrometer ($\delta_r$).  Thus, for a given neuron, the total delay in the signal coming from the left side is $\tau_{\ell} + a\delta_{\ell}$, and the total delay in the signal coming from the right side is $\tau_r+b\delta_r$ for integers $a$ and $b$ denoting how many neurons into the column the neuron is from the dorsal and ventral surface respectively.  A neuron receives coincident signals and fires whenever $|(\tau_{\ell} + a\delta_{\ell})-(\tau_r+b\delta_r)| <\varepsilon$ for some error bound $\varepsilon$.   

Notice that because a sound wave of a given frequency is periodic with period $T$, the signals will also be coincident for $|(\tau_{\ell} + a\delta_{\ell})-(\tau_r+b\delta_r)|\mod T <\varepsilon$, and so neurons will also fire in response to time differences that are integer multiples of the period away from the true time difference \cite{CarrAxonal}. Also observe that the time difference in the signals to each neuron in a column changes at a constant rate, $\delta_{\ell} + \delta_r$ per neuron, which implies that a neuron in the column fires every $\frac{T}{\delta_l + \delta_r}$ neurons.  This gives rise to periodicity in the column.  Behavioral experiments show that when localizing pure tones, owls may make errors in sound localization by responding to phantom targets, responding to the location of a sound at one of the multiples of the period rather than the location of the true time difference, showing the ambiguity of time difference as a sound localization cue due to the periodic nature of sound waves\cite{KonishiOwl}.

Thus, we have seen that periodic codes arise biologically, so we now ask what behavioral implications such a code has.  The periodic code in the owl's nucleus laminaris is part of the system that the owl uses to determine the position of a sound source on the horizon.  It is natural to ask how this code relates to a more highly studied position code, the \textit{place code} in the place cells of the mammalian hippocampus.  Each of these place cells fire over a convex set corresponding to the animal's position in the environment.  We consider whether the cells in the owl's nucleus laminaris can be associated to convex subsets of angles on the horizon, addressed formally in the next section.

\section{Convex closures of periodic codes}
Inspired by the periodic structure of the neural code in the nucleus laminaris of the barn owl, we explore the convexity of periodic codes, beginning by formally defining a convex code and introducing the concept of a \textit{convex closure}.  By considering the biological relevance of these concepts, we demonstrate that convexity is important to the owl's sound localization ability.  We then present our main result, Theorem \ref{subcompletionTheorem}, and conclude by proving it.

\subsection{Convex codes and the convex closure}
Here we review the concept of \textit{convex codes} and some basic results \cite{CurtoConvex, CurtoArticle} before introducing a new concept, the \textit{convex closure}.

Given an \textit{open cover} $\mathcal{U}$ of a topological space $X$, where $\mathcal{U}$ is the collection of open sets $\{U_1, \ldots, U_n\}$ such that $U_i \subset X$, we can define a \textit{code of the cover} $C(\mathcal{U})$,  
$$C(\mathcal{U}) \od \left\{\sigma \subseteq [n] \mid \bigcap_{i\in \sigma}U_i \setminus \bigcup_{j\in [n]\setminus \sigma}U_j \neq \emptyset\right\}.$$
In $C(\mathcal{U})$, each $U_i$ is called the {\it receptive field} of neuron $i$.  
We say that a  code is \textit{convex} if it can be realized as $C(\mathcal{U})$ where each of the $U_i$ is an open convex set.  As an example, see Figure~\ref{codeEx}.  Note that not every code is convex because there are geometric and topological constraints imposed by convexity \cite{CurtoConvex}.

\begin{figure}
	\centering
		\includegraphics[width=.4\textwidth]{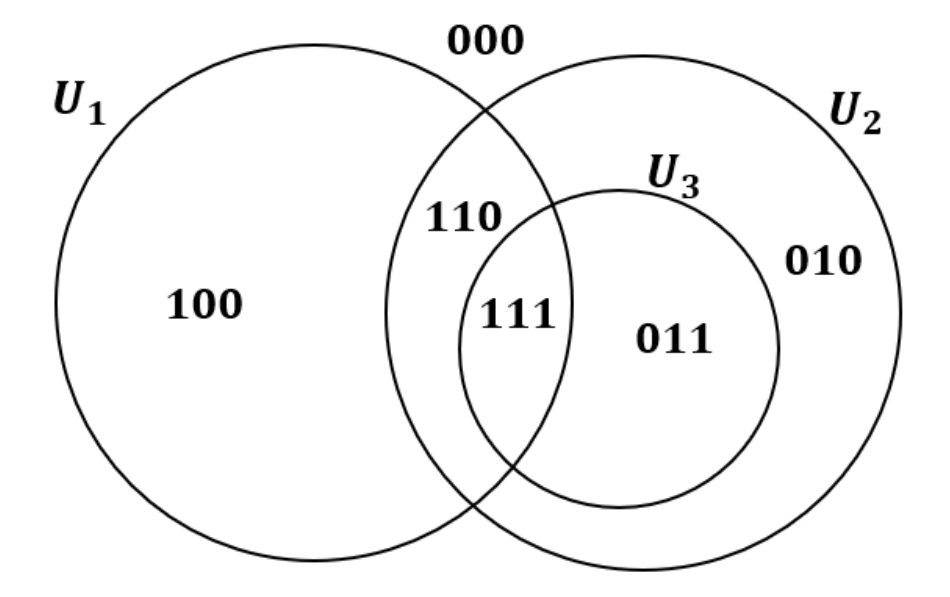}
		\caption{\footnotesize \emph{A code of a cover, $C(\mathcal{U})$}. This arrangement of receptive fields $\mathcal{U} = \{U_1, U_2, U_3\}$ corresponds to $C(\mathcal{U}) = \{100, 110, 111, 010, 011, 000\}$.  Neurons which cofire correspond to a region of intersection of their corresponding receptive fields.  Observe that $C(\mathcal{U})$ is a convex code since each $U_i$ is a convex set.}
	\label{codeEx}
\end{figure}

An important property that prevents a code from having a convex realization is based on the links of the simplicial complex of the code.  
The \textit{link} of a face $\sigma$ in a simplicial complex $\Delta$, denoted $\Lk_{\sigma}(\Delta)$, is
$$\Lk_{\sigma}(\Delta) = \{\omega \in \Delta \mid \sigma \cap \omega = \emptyset \text{ and } \sigma \cup \omega \in \Delta\}.$$
To every $\Delta$ we associate a unique \textit{minimal code} consisting of all $\sigma$ having non-contractible links:
\[C_{\min}(\Delta) = \{\sigma \in \Delta \mid \Lk_{\sigma}(\Delta) \text{ is non-contractible}\}. \]

\begin{figure}
	\centering
		\includegraphics[width=.3\textwidth]{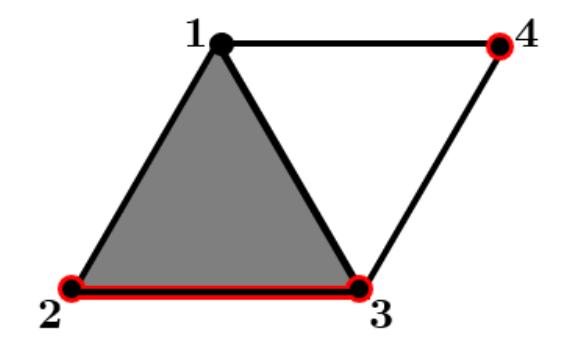}
		\caption{\footnotesize \emph{A link in a simplicial complex, $\Lk_{\{1\}}(\Delta)$}. The simplicial complex $\Delta$ is shown in black and gray.  The link of vertex 1, $\Lk_{\{1\}}(\Delta) = \{\emptyset, \{2\}, \{3\}, \{4\}, \{2,3\}\}$ is highlighted in red, and we see it is disconnected and hence non-contractible.}
	\label{linkEx}
\end{figure}

An example of these concepts for a simplicial complex $\Delta$ is shown in Figure~\ref{linkEx}, where we see $\Lk_{\{1\}}(\Delta)$ is non-contractible.  As a result, $1000 \in C_{\min}(\Delta)$.  In contrast, $\Lk_{\{2\}}(\Delta) = \{\emptyset, \{1\}, \{3\}, \{1, 3\}\}$, which is contractible, so $0100 \notin C_{\min}(\Delta)$.    

If $\tau$ is a facet of $\Delta$, then $\Lk_{\tau}(\Delta)  = \emptyset$, which is non-contractible. It follows that all facets of $\Delta$ are automatically contained in $C_{\min}(\Delta)$. In fact, these facets correspond to the maximal codewords of any code with simplicial complex $\Delta$ (see \cite{CurtoConvex} for more details).
Note that in case of periodic codes, $C_{k,m}(n)$, every codeword is maximal and corresponds to a facet of $\Delta = \Delta(C_{k,m}(n))$. Therefore, for periodic codes, we always have 
$C_{k,m}(n) \subseteq C_{\min}(\Delta).$

We call the elements of $C_{\min}(\Delta)$ \textit{mandatory codewords} because they must all be included in any convex code $C$ with simplicial complex $\Delta$. This follows from \cite[Theorem 1.3]{CurtoConvex}, with the relevant portion summarized in the lemma below.

\begin{lemma}
Let $C$ be a code with simplicial complex $\Delta$.
If $C \not \supseteq C_{\min}(\Delta)$, then $C$ is not a convex code.
\label{mandCodewords}
\end{lemma}

A counterexample given in \cite{counterexample} illustrates that the converse is not true; a code may contain all the mandatory codewords but may still not have a convex realization.  To address this, we introduce the concept of a convex closure.

\begin{definition}
{A \textit{convex closure} $\bar{C}$ of $C$ is a convex code of smallest size such that $C \subseteq \bar{C}$ and $\Delta(\bar{C}) = \Delta(C)$.} 
\end{definition}

We note that the convex closure is a closure operator on the power set $\textbf{P}([n])$.\footnote{A closure operator, $Cl: \textbf{P}(S) \rightarrow \textbf{P}(S)$, maps the power set of $S$ to itself and for $X, Y \subseteq S$ satisfies i) $X\subseteq Cl(X)$, ii) if $X\subseteq Y$, then $Cl(X) \subseteq Cl(Y)$, and iii) $Cl(Cl(X)) = Cl(X)$.}  From Lemma \ref{mandCodewords}, it is clear that $C_{\min}(\Delta(C)) \subseteq \bar{C}$, but there are cases where $\bar{C}$ must contain additional codewords \cite{counterexample}.

\subsection{Convexity and sound localization}

As we will see in the next section, periodic codes are not generally convex, except in degenerate cases. However, because of the various advantages of convex codes in associating firing patterns with a specific region of the stimulus space, we are interested in how we may modify periodic codes to attain convexity. 

Convexity is especially relevant to the periodic codes in the nucleus laminaris (NL) of the owl as the function of this brain structure is to locate sounds on the horizon, which is equivalent to determining the convex set of angles from which the sound originated.  When receiving a pure tone, the owl makes predictable errors in its judgment of the sound's location.  The phantom targets to which the owl responds are not random but correspond to the location of a sound source where the time difference reaching the ears is the true time difference plus some multiple of the period of the sound wave (see Figure \ref{AuditorySystem}b). This suggests that the owl is able to localize a sound to a choice of several disconnected sets of angles, rather than a single convex set.  This behavior corresponds to the fact that the neural code in the NL does not have a convex realization as we will show in Theorem \ref{subcompletionTheorem}.

However, such behavioral errors are rare and are restricted to the case of single frequency tones.  When responding to wide bandwidth sounds, the owl's average error in sound localization is one third of its average error in responding to a single frequency tone \cite{Knudsen1}.  Higher in the brain stem, biologists have observed space mapped cells in the external nucleus of the inferior colliculus, which receives inputs from multiple frequency columns.  This suggests that the biological system somehow forms a convex code from the nonconvex periodic codes in the NL, a biological convex closure, so that the owl is able to locate wide bandwidth sounds within an average of 2 degrees \cite{Knudsen2}.  Such a convex code may arise by combining the code in multiple isofrequency columns of the NL or through stochasticity in neural firing. We explore these possibilities in the combinatorial neural code framework, answering the following questions.

\bigskip

\textbf{Question 1.} What codes are a convex closure of a periodic code?  

\bigskip

We answer this question in Theorem \ref{subcompletionTheorem}, showing that the convex closure is unique.  Given this convex closure, we then analyze how this convex closure could arise, asking the following question.

\bigskip

\textbf{Question 2.} How does stochastic noise in the firing patterns of a $\km$ periodic code alter the convexity of the code?

\bigskip

We explore this question in Section 4.

\subsection{The convex closure of a periodic code}

Our main result in this section is Theorem~\ref{subcompletionTheorem}, which gives the convex closure of any $\km$ periodic code for $k \leq m$.  We restrict our analysis to the case where $k\leq m$, requiring a certain degree of sparsity in the code.  Such sparse codes better reflect the codes which arise biologically and decrease the number of nontrivial intersections among neurons.  

\begin{theorem}Let $C = C_{k,m}(n)$ be a $\km$ periodic code on $n$ neurons with simplicial complex $\Delta$.  
For $k\leq m$, the convex closure of $C$ is precisely $\bar{C} = C_{\min}(\Delta)$, and is thus unique. Moreover,
\begin{enumerate}
\item $\bar{C} = C_{\min}(\Delta) = C$ if $k=0$ or $k=1$, and
\item $\bar{C} = C_{\min}(\Delta) = C_{k,m}(n) \cup C_{k-1, m+1}(n)$ if $1<k\leq m$.
\end{enumerate}
\label{subcompletionTheorem}
\end{theorem}

As an example, consider $C_{2,3}(5)$, which contains the codewords 11000 and 01100.  The codeword 11000 implies that $U_1 \cap U_2 \neq \emptyset$.  Similarly, the codeword 01100 implies that $U_2 \cap U_3 \neq \emptyset$.  There are no codewords for which neuron 1 and neuron 3 cofire, so $U_1 \cap U_3 = \emptyset$.  However, the codeword 01000 is not in the code, so $U_2$ is entirely contained in $U_1 \cup U_3$.  Since $U_1$ and $U_3$ are disjoint, this can only be true if $U_2$ is disconnected, and hence not convex.  As illustrated in Figure \ref{ConvexTheorems}, the code $\bar{C} = C_{2,3}(5) \cup C_{1,4}(5)$ has a convex realization and $\Delta(\bar{C}) = \Delta(C)$.

\begin{figure}
	\centering
		\includegraphics[width=.9\textwidth]{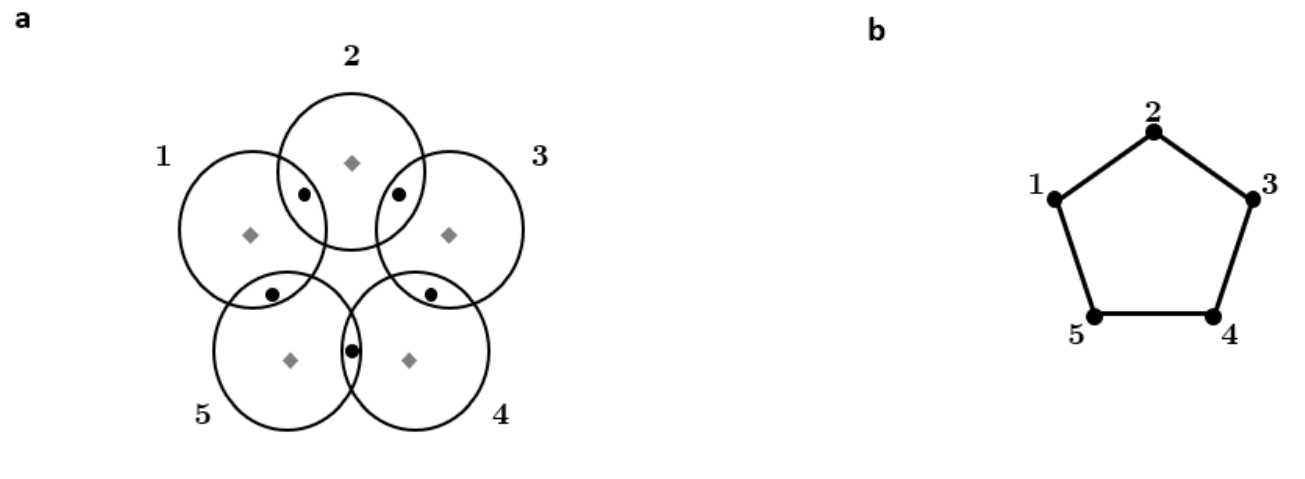}
		\caption{\footnotesize \emph{Convex Closure of $C_{2,3}(5)$}.  \textbf{(a)} The convex closure of $C_{2,3}(5)$ is $\bar{C} = C_{2,3}(5) \cup C_{1,4}(5)$ (Theorem \ref{subcompletionTheorem}), which has a convex realization as shown here.  The original code, $C_{2,3}(5)$ (black dots) has a convex realization with the union of $C_{1,4}(5)$ (gray diamonds).  \textbf{(b)}  The convex closure preserves the simplicial complex, $\Delta(\bar{C}) = \Delta(C_{2,3}(5))$, shown here.}
	\label{ConvexTheorems}
\end{figure}

This result is particularly interesting because it gives an example of a class of codes where $\bar{C} = C_{\min}(\Delta(C))$.  In general, we do not always have $C \subset C_{\min}(\Delta(C))$, but we do in the case of periodic codes because every codeword corresponds to a facet of the simplicial complex.  What we also see in these codes is that convexity is fully determined by containing $C_{\min}(\Delta)$, which is not always true \cite{counterexample}.  This raises the question of whether there are special properties of $\Delta(C_{k,m}(n))$ which can be used to detect more generally when $C_{\min}(\Delta)$ is convex for some $\Delta$.

To prove Theorem~\ref{subcompletionTheorem}, we need the following two propositions:

\begin{proposition}
Let $\Delta = \Delta(C_{k,m}(n))$ and $1<k\leq m$.  Then $C_{k,m}(n) \cup C_{k-1, m+1}(n) \subseteq C_{\min}(\Delta)$.
\label{subcodeMandatory}
\end{proposition}

\begin{proposition}
The code $C_{k,m}(n) \cup C_{k-1, m+1}(n)$ is convex for $1<k \leq m$.
\label{convexRealization}
\end{proposition}

Given these two propositions, which we will prove in the next subsection, we can now prove Theorem~\ref{subcompletionTheorem}.

\begin{proof}[Proof of Theorem~\ref{subcompletionTheorem}]
Let $C = C_{k,m}(n)$ for $k\leq m$ and $\Delta = \Delta(C)$.  By definition, $\bar{C} \supseteq C$, $\Delta(\bar{C}) = \Delta$, and $\bar{C}$ is convex. By Lemma~\ref{mandCodewords}, we also have $\bar{C} \supseteq C_{\min}(\Delta)$.  Recall that since every codeword in $C$ corresponds to a facet of $\Delta$, we also have $C \subseteq C_{\min}(\Delta)$ and thus,
$$C \subseteq C_{\min}(\Delta) \subseteq \bar{C}.$$ 
To show that $\bar{C} = C_{\min}(\Delta)$, it thus suffices to show that  $C_{\min}(\Delta)$ is convex.  

In the cases $k=0$ and $k=1$, we can see directly that $C$ is convex, and thus $\bar{C} = C_{\min}(\Delta) = C$. For $k=0$, $C$ consists of only the all-zeros codeword, $00\cdots0$, and is thus trivially convex. For $k=1$, all codewords in $C$ are disjoint and, since they all correspond to facets of $\Delta$, we see that all facets of $\Delta$ are disjoint. It then follows from \cite[Proposition 2.6]{CurtoConvex} that $C$ is convex.

For the remaining cases, $1 < k \leq m$, we have that $C_{k,m}(n) \cup C_{k-1, m+1}(n) \subseteq C_{\min}(\Delta)$ (Proposition~\ref{subcodeMandatory}), and thus
$$C \subseteq  C_{k,m}(n) \cup C_{k-1, m+1}(n) \subseteq C_{\min}(\Delta) \subseteq \bar{C}.$$
We also have that $C_{k,m}(n) \cup C_{k-1, m+1}(n)$ is convex (Proposition~\ref{convexRealization}), which immediately implies $\bar{C} = C_{\min}(\Delta) = C_{k,m}(n) \cup C_{k-1, m+1}(n)$, as desired.
\end{proof}

\subsubsection*{Biological relevance}
Theorem~\ref{subcompletionTheorem} implies periodic codes are not convex except in the cases where $k=0$ and $k=1$, corresponding respectively to a state of constant inactivity or perfect precision. Both extremes are unlikely given the inherent stochasticity in this system (see Section 4 for more details).  The second part of the theorem, for $k>1$, can provide insight into the behavioral errors made by the owl in locating a single frequency tone: because the neural code for a single isofrequency column of the nucleus laminaris (NL) is not convex, the owl cannot determine a unique open region in space from which the sound must have originated.

Yet, the owl does not always make these behavioral errors, suggesting that the biology of the system provides a way for the owl to disambiguate the possible locations of the sound source, thus yielding a convex neural code. Theorem~\ref{subcompletionTheorem} shows that a convex closure of a periodic code is formed from the union of two periodic codes.  In particular, these two codes differ by whether a single neuron is firing in each firing band; the close relationship between these codes suggests stochasticity in the system could result in the convex closure, which we address in greater depth in Section 4.  

In addition to the possible stochastic relationship between the two codes, the two periodic codes may also be related by the connections among nuclei in the barn owl's brain stem.  The NL projects to the central nucleus of the inferior colliculus (IC), which sends inputs from multiple isofrequency laminae to the external nucleus of the IC \cite{WagnerOwl}.  Recall that each isofrequency column $i$ of the NL with frequency $\frac{1}{T_i}$ fires every $\frac{T_i}{\delta_l+\delta_r}$ neurons (Figure~\ref{AuditorySystem}b), which corresponds to a different $m$ value in the periodic code of each column.  Thus, receiving input from multiple columns is analogous to receiving input from different periodic codes, as is needed for the convex closure of a single column; in fact, it has been observed that the first space mapped cells exist in the external nucleus of the IC, where the signals for multiple frequencies first converge \cite{WagnerOwl}.

\subsection{Proof of Propositions~\ref{subcodeMandatory} and \ref{convexRealization}}

\subsubsection*{Proof of Proposition~\ref{subcodeMandatory}}
As has already been noted, since every codeword in $C_{k,m}(n)$ corresponds to a facet, $C_{k,m}(n) \subseteq C_{\min}(\Delta)$.  To prove Proposition~\ref{subcodeMandatory}, we will show that for any $\tau \in C_{k-1, m+1}(n)$, the link $\Lk_{\tau}(\Delta)$ is non-contractible, and so $\tau \in C_{\min}(\Delta)$. We first consider the special case of $\km$ periodic codes on $k+m$ neurons (Lemma~\ref{subbandLink}) and then extend this result to a code on $n$ neurons using Proposition~\ref{linkLemma}.

We introduce the notation $\sigma_{i,j}(n)$ for $i,j \in [n]$ to be the face of a simplicial complex on a set of consecutive vertices, where we consider the $n$th and first vertex to be adjacent.  More formally, we define  

$$\sigma_{i,j}(n) = \begin{cases}
\{\ell \mid i \leq \ell \leq j\} & \text{if } i\leq j\\
\{\ell \mid i \leq \ell \leq n\} \cup \{\ell \mid 1\leq \ell \leq j\} & \text{if } i>j
\end{cases}
$$
  
Note that a general simplicial complex may not contain such a face, but we are interested specifically in $\Delta(C_{k,m}(n))$, which contains $\sigma_{i,j}(n)$ whenever $|\sigma_{i,j}(n)| \leq k$.  In the case $n=k+m$, the collection of all $\sigma_{i,j}(k+m)$ with $|\sigma_{i,j}(k+m)| = x$ for $x \leq k$ corresponds to $C_{x, k+m-x}(k+m)$.

\begin{lemma}
Let $\Delta = \Delta(C_{k,m}(k+m))$ and $m \neq 0$.  If $|\sigma_{i,j}(k+m)| = k-1$, then $\Lk_{\sigma_{i,j}(k+m)}(\Delta)$ is non-contractible.
\label{subbandLink}
\end{lemma} 

\begin{proof}
Let $\sigma = \sigma_{i,j}(k+m)$ where $j = (i+k-2)\mod{(k+m)}$.  We have $\Lk_{\sigma}(\Delta) = \{h, \ell\}$ where $h=i-1\mod{(k+m)}$ and $\ell = i+k-1\mod{(k+m)}$.  These vertices are distinct since they are distance $k$ apart and $m\neq 0$, so the set is disconnected and non-contractible.
\end{proof}

This lemma allows us to show that $C_{k,m}(n)$ restricted to the first $k+m$ vertices has non-contractible link for the faces corresponding to $C_{k-1, m+1}(k+m)$.  We want to show that every $\tau$ corresponding to a codeword $C_{k-1, m+1}(n)$ has non-contractible link in $\Delta(C_{k,m}(n))$.  To do so, we use the following lemma, which requires the notation for the restricted simplicial complex,
$$\Delta\vert_{\sigma\cup\tau} = \{ s \in \Delta \mid s \subseteq \sigma \cup \tau \}.$$

\begin{lemma}
\cite[Corollary~4.3]{CurtoConvex} Suppose $v \notin \sigma$ and $\sigma \cap \tau = \emptyset$.  If $\Lk_\sigma(\Delta\vert_{\sigma\cup\tau})$ is non-contractible, then $\Lk_\sigma(\Delta\vert_{\sigma\cup\tau\cup v})$, $\Lk_{\sigma \cup v}(\Delta\vert_{\sigma\cup\tau\cup v})$, or both are non-contractible. 
\label{CurtoLink}
\end{lemma}

One way to demonstrate contractibility is to show that a simplicial complex is a cone.  We will use the fact that a simplicial complex is a cone if and only if the intersection of all the facets of the simplicial complex is nontrivial in the proof of the following proposition.

\begin{proposition}
Let $\Delta = \Delta(C_{k,m}(n))$.  Let $\tau = \{v_i \mid 1 \leq i \leq k+m\}$, $\sigma \subset \tau$, and $j>k+m$.  Suppose $\Lk_\sigma(\Delta\vert_\tau)$ is non-contractible:
\begin{enumerate}
\item If $j\equiv i \mod{(k+m)}$ for some $i$ such that $v_i\in \sigma$, then $\Lk_{\sigma \cup v_j}(\Delta \vert_{\tau \cup v_j})$ is non-contractible.
\item Otherwise, $\Lk_{\sigma}(\Delta \vert_{\tau \cup v_j})$ is non-contractible.
\end{enumerate}
\label{linkLemma}
\end{proposition}

\begin{proof}
1.  Assume $j \equiv i \mod{(k+m)}$ for some $i$ such that $v_i\in \sigma$.  To show that $\Lk_{\sigma \cup v_j}(\Delta \vert_{\tau \cup v_j})$ is non-contractible, it suffices to show that $\Lk_\sigma(\Delta\vert_{\tau \cup v_j})$ is contractible (Lemma \ref{CurtoLink}).  By Proposition \ref{simplexStructure}, for any $\omega \subset \tau$, if $v_i \cup \omega \in \Delta$, then $v_i\cup v_j \cup \omega \in \Delta$.  Thus, every facet of $\Lk_{\sigma}(\Delta\vert_{\tau\cup v_j})$ contains $v_j$.  It follows that $\Lk_{\sigma}(\Delta \vert_{\tau \cup v_j})$ is a cone and hence contractible.  2.  Assume $j\not\equiv i \mod{(k+m)}$ for any $i$ such that $v_i \in \sigma$.  Then $j \equiv \ell \mod{(k+m)}$ for some $\ell$ such that $v_\ell \in \tau \setminus \sigma$.  We have two cases.  Case 1: Assume that $v_\ell \notin \Lk_\sigma(\Delta\vert_\tau)$.  This implies $v_\ell \cup \sigma \notin \Delta \vert_\tau$, and so $v_\ell\cup \sigma \notin \Delta \vert_{\tau\cup v_j}$; by Proposition \ref{simplexStructure}, $v_j \cup \sigma \notin \Delta\vert_{\tau \cup v_j}$.  This implies $v_j \notin \Lk_\sigma(\Delta\vert_{\tau \cup v_j})$, and we have $\Lk_\sigma(\Delta \vert_{\tau\cup v_j}) = \Lk_\sigma(\Delta \vert_\tau)$, which is non-contractible.  Case 2: Assume that $v_\ell \in \Lk_\sigma(\Delta\vert_\tau)$.  Let $\Lambda$ be the closure of the set $\{\rho \in \Lk_\sigma(\Delta\vert_\tau) \mid v_\ell \subseteq \rho\}$ -- that is, the smallest simplicial complex containing the faces of $\Lk_\sigma(\Delta \vert_\tau)$ that contain $v_\ell$.  By Proposition \ref{simplexStructure}, for any $\omega \in \Delta \vert_{\tau \cup v_j}$, if $v_\ell \cup \omega \in \Delta\vert_{\tau \cup v_j}$ then $v_j \cup v_\ell \cup \omega \in \Delta\vert_{\tau \cup v_j}$.  This implies that $\Lk_\sigma(\Delta\vert_{\tau\cup v_j}) = \Lk_\sigma(\Delta \vert_\tau) \cup \cone_{v_j}(\Lambda)$, where $\cone_{v_j}(\Lambda) = \{\gamma \cup v_j \mid \gamma \in \Lambda\}$.  Observe that $\Lambda$ itself is a cone since all the facets contain $v_\ell$, and hence we have coned off a contractible subcomplex of $\Lk_\sigma(\Delta \vert_\tau)$.  This implies that the homotopy type of $\Lk_\sigma(\Delta\vert_{\tau\cup v_j})$ is the same as the homotopy type of $\Lk_\sigma(\Delta \vert_\tau)$, which is non-contractible.
\end{proof}

We are now able to extend the results of Lemma \ref{subbandLink} to a code on $n$ neurons, completing the proof of Proposition~\ref{subcodeMandatory}.

\begin{proof}[Proof of Proposition~\ref{subcodeMandatory}]
By Lemma \ref{subbandLink}, the link of every $\sigma_{i,j}(k+m)$ with $|\sigma_{i,j}(k+m)| = k-1$ is non-contractible in $\Delta(C_{k,m}(k+m))$.  Observe that the codewords corresponding to $\sigma_{i,j}(k+m)$ of dimension $k-1$ are the codewords of $C_{k-1, m+1}(n)$ restricted to the first $k+m$ neurons.  Each codeword in $C_{k-1,m+1}(n)$ is the set of neurons which are equivalent modulo $k+m$ to some $\sigma_{i,j}(k+m)$.  Let $\omega$ be a face corresponding to a codeword in $C_{k-1, m+1}(n)$ with corresponding $\sigma' = \sigma_{i,j}(k+m)$.  By Proposition \ref{linkLemma}, since $\Lk_{\sigma'}(\Delta(C_{k,m}(n)))$ is non-contractible, $\Lk_\omega(\Delta(C_{k,m}(n)))$ is non-contractible.  Thus, the link of the face corresponding to every codeword in $C_{k-1, m+1}(n)$ is non-contractible, so $C_{k-1, m+1}(n) \subset C_{\min}(\Delta)$. 
\end{proof}

\subsubsection*{Proof of Proposition~\ref{convexRealization}}
To prove Proposition~\ref{convexRealization}, we need to show $C_{k,m}(n) \cup C_{k-1, m+1}(n)$ has a convex realization.  We first explicitly construct a convex realization for $n = k+m$ (Proposition~\ref{convexReal}) and then show that this realization can be extended to a code on $n$ neurons (Lemma~\ref{convexityDependence}).

We first begin with a construction which will help us construct a convex realization of $C_{k,m}(k+m) \cup C_{k-1, m+1}(k+m)$.

\begin{definition}
A \textit{circular cover}, $\Gamma$, is a collection of open arcs $\{\gamma_i\}$ with $\gamma_i = \gamma(a, b)$ for $-2\pi \leq a < b\leq 2\pi$ and $b-a\leq 2\pi$, where $\gamma(a, b) \od \{(\cos{\theta}, \sin{\theta})\mid a < \theta < b\}$.
\end{definition}

Note that this cover is circular in the sense that it is composed of circular arcs, but the definition does not require that the union of the arcs covers $S^1$.  Analogously to the way we defined the code of a cover, we can define the code of a circular cover.

\begin{definition}The \textit{code of a circular cover} $\Gamma = \{\gamma_i\}$ is the neural code 
$$C(\Gamma) \od \left\{\sigma \subseteq [n] \mid \bigcap_{i\in \sigma} \gamma_i \setminus \cup_{j \in [n] \setminus \sigma} \gamma_j \neq \emptyset\right\}.$$ 
\end{definition}

\begin{figure}
	\centering
		\includegraphics[width = .9\textwidth]{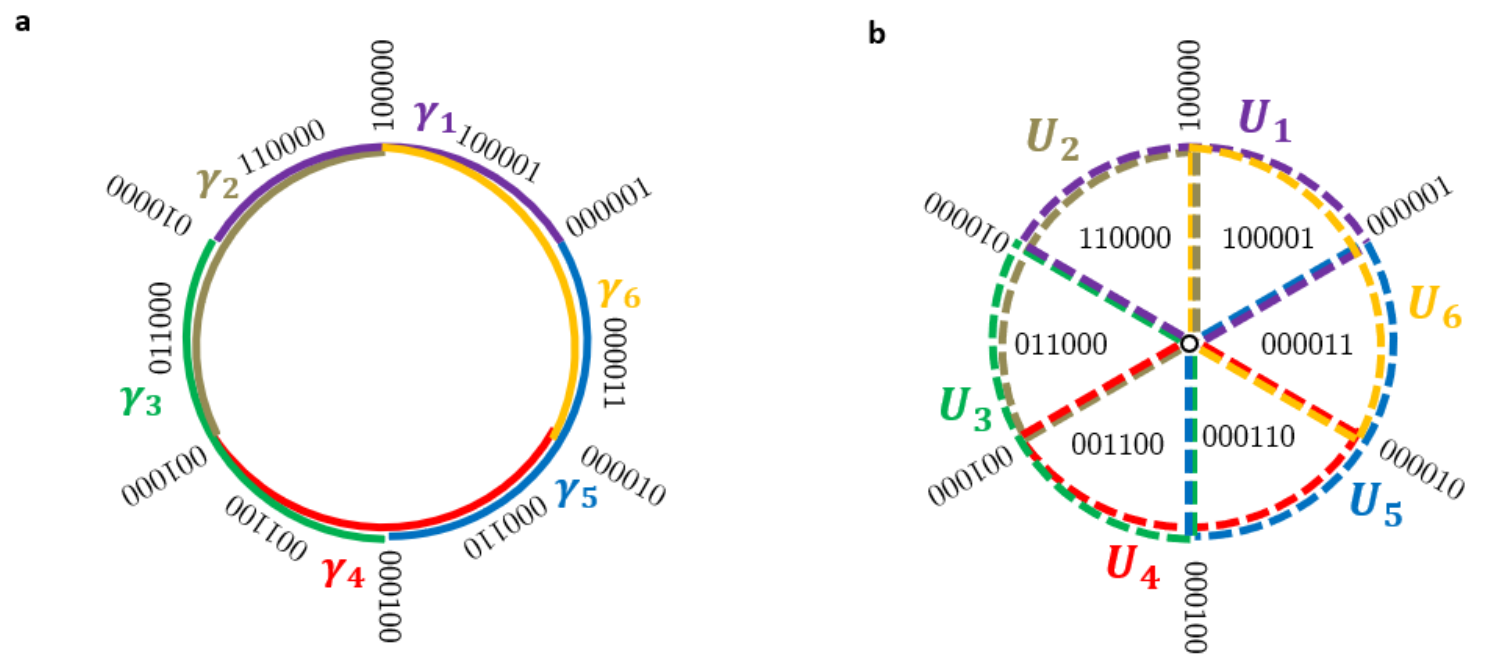}
		\caption{\footnotesize \textit{Circular and Convex Realizations of $\bar{C} = C_{2,4}(6) \cup C_{1,5}(6)$}.  \textbf{(a)} A circular realization of $\bar{C}$.  Each colored open arc, $\gamma_i$ is an element of $\Gamma$ and can be identified with an open sector $U_i$ of the same color in $C(\mathcal{U})$. \textbf{(b)} A convex realization of $\bar{C}$, derived from the circular realization.}
	\label{CircRealization}
\end{figure}

An example of the code of a circular cover is shown in Figure~\ref{CircRealization}.  Observe that any point $p$ written in polar coordinates as $(r, \theta)$ with $r=1$ corresponds to a codeword $c_1 \cdots c_n \in C(\Gamma)$ where $c_i = 1$ if $(\cos{\theta},\sin{\theta}) \in \gamma_i$ and $c_i = 0$ otherwise.  To prove when such circular realizations are convex, we introduce the \textit{support} of a codeword $c$, denoted $\supp(c)$ as the set $\{i\mid c_i = 1\}$.

\begin{proposition} \label{circConvex}
If $\len(\gamma_i)\leq \pi$ for all $\gamma_i \in \Gamma$, then $C(\Gamma)$ has a convex realization.
\end{proposition}
\begin{proof}
Let $\Gamma$ be a circular cover.  Define an open cover $\mathcal{U} = \{U_i\}$ in $\reals^2$ by the following method.  For each open arc $\gamma_i = \gamma(a_i, b_i) \in \Gamma$, let $U_i$ be the open sector $\{(r\cos{\theta}, r\sin{\theta}) \mid 0<r<1 $ and $a_i<\theta<b_i\}$.  Since $b_i - a_i = \len(\gamma_i) \leq \pi$, each $U_i$ is convex.  Observe that the origin is not contained in any $U_i$.  We want to show that $C(\mathcal{U}) = C(\Gamma)$ and thus $C(\Gamma)$ has a convex realization.  Suppose $c \in C(\mathcal{U})$.  There exists some point $p \in \cap_{i\in \supp(c)}U_i$, which can be written as $(r\cos{\theta},r \sin{\theta})$.  Map $p$ to the corresponding point $p_{\theta} = (\cos{\theta}, \sin{\theta})$ on the unit circle.  We have $p_{\theta} \in \cap_{i\in \supp(c)} \gamma_i$, and thus $c \in C(\Gamma)$ also.  This implies $C(\mathcal{U}) \subseteq C(\Gamma)$.  To see that $C(\Gamma) \subseteq C(\mathcal{U})$, let $c \in C(\Gamma)$.  There exists some point $p_{\theta} \in \cap_{i\in \supp(c)}\gamma_i$.  Write $p_{\theta}$ as $(\cos{\theta}, \sin{\theta})$.  Map $p_{\theta}$ to the point $p = (\frac{1}{2} \cos{\theta}, \frac{1}{2}\sin{\theta})$.  We have $p \in \cap_{i \in \supp(c)}U_i$, and thus $c \in C(\mathcal{U})$.  This gives $C(\Gamma) \subseteq C(\mathcal{U})$, and so $C(\Gamma) = C(U)$. 
\end{proof}

Figure~\ref{CircRealization} shows an example of $C(U) = C(\Gamma)$ and the correspondence between the open arcs and open sectors.  We can now show that the code $C_{k,m}(k+m) \cup C_{k-1, m+1}(k+m)$ is convex by showing that it arises from a circular cover and applying Proposition~\ref{circConvex}.

\begin{proposition} \label{convexReal}
The code $C_{k,m}(k+m) \cup C_{k-1, m+1}(k+m)$ has a convex realization for $k \leq m$.
\end{proposition}
\begin{proof}
We define an open arc $\gamma_i$ for each neuron $i$ in the following manner.  If $i\leq m$, let $\gamma_i = \gamma(i \frac{2\pi}{k+m}, (i+k) \frac{2\pi}{k+m})$.  If $i>m$, let $\gamma_i = \gamma(i \frac{2\pi}{k+m} - 2\pi, (i+k) \frac{2\pi}{k+m}-2\pi)$.  This collection of receptive fields gives a circular cover, $\Gamma$, with $\len(\gamma_i) = k\frac{2\pi}{k+m} \leq \pi$ for all $i$.  By Proposition \ref{circConvex}, $C(\Gamma)$ is convex.  We need only show that $C(\Gamma) = C_{k,m}(k+m) \cup C_{k-1, m+1}(k+m)$.  At every angle which is not a multiple of $\frac{2\pi}{k+m}$, exactly $k$ of the $\gamma_i$ intersect, corresponding to codewords in $C_{k,m}(k+m)$.  At angles that are a multiple of $\frac{2\pi}{k+m}$, exactly $k-1$ of the $\gamma_i$ intersect, corresponding to codewords in $C_{k-1, m+1}(k+m)$.  Thus, $C(\Gamma) =C_{k,m}(n) \cup C_{k-1, m+1}(k+m)$.
\end{proof}

We want to extend this convex realization on $k+m$ neurons to a code on $n$ neurons.

\begin{lemma}{Let $C$ be a code where for every codeword $c_1\cdots c_n \in C$, if $i \equiv j \mod{z}$, then $c_i = c_j$.  The code $C$ is convex if and only if $C$ restricted to any $z$ consecutive neurons is convex.}
\label{convexityDependence}
\end{lemma}

\begin{proof}
($\Rightarrow$)  Suppose $C$ is convex.  There exists a convex, open cover $\mathcal{U} = \{U_i\}_{i\in[n]}$ such that $C(\mathcal{U}) = C$.  Let $C \vert_Z$ be $C$ restricted to a set $Z$ consisting of $z$ consecutive neurons.  Then $\mathcal{U}'=\{U_i\}_{i\in Z}$ is an open cover with $C(\mathcal{U}) = C \vert_Z$, so $C \vert_Z$ is convex.  ($\Leftarrow$)  Suppose $C \vert_Z$ is convex on any set $Z$ of $z$ consecutive neurons.  Without loss of generality, there exists a convex, open cover $\mathcal{U} = \{U_i\}_{1\leq i\leq z}$.  Define a convex, open cover $\mathcal{U}'=\{U_j \mid U_j = U_i$ if $i \equiv j\mod{z}\}$.  By assumption, if $i \equiv j\mod{z}$, then $c_i = c_j$ in every codeword in $C$, so $C(\mathcal{U}') = C$.  Therefore, $C$ has a convex realization.  
\end{proof}

Lemma~\ref{convexityDependence} implies that the convexity of a $\km$ periodic code depends only on whether the code can be realized convexly on the first $k+m$ neurons.  We can now complete the proof of Proposition~\ref{convexRealization}.

\begin{proof}[Proof of Proposition~\ref{convexRealization}]
The code $C_{k,m}(n) \cup C_{k-1, m+1}(n)$ has the property that $c_i = c_j$ if $i \equiv j \mod{(k+m)}$ for every codeword.  By Lemma \ref{convexityDependence}, since $C_{k,m}(k+m) \cup C_{k-1, m+1}(k+m)$ has a convex realization for $k \leq m$ (Proposition~\ref{convexReal}), $C_{k,m}(n) \cup C_{k-1, m+1}(n)$ has a convex realization for $k \leq m$.
\end{proof}

Observe that for any $k$ and $m$, we can construct a circular realization of the code by defining open arcs in the same construction as above.  However, for $k>m$, the method of constructing sectors as open sets no longer generates convex sets because we do not have $\len(\gamma_i) \leq \pi$.  The question of how and whether a convex realization can be constructed for $C_{k,m}(n)\cup C_{k-1, m+1}(n)$ and $k>m$ remains open, as some but not all of these codes do not contain $C_{\min}(\Delta(C))$.  For example, let $C' = C_{3,2}(5) \cup C_{2,3}(5)$ and $C'' = C_{4,1}(5) \cup C_{3,2}(5)$.  We have $C' \supseteq C_{\min}(\Delta(C'))$.  On the other hand, $C'' \not \supseteq C_{\min}(\Delta(C''))$, as $\Lk_{1,4}(\Delta) = \{23, 35, 25\}$, which is noncontractible.

\section{Stochastically convex periodic codes}
Here we address Question 2 from Section 3.2.  In Theorem~\ref{subcompletionTheorem}, we showed that for $1<k\leq m$, the periodic code $C_{k,m}(n)$ can be completed to a convex code by adding the set of codewords in $C_{k-1,m+1}(n)$.  One of the possible ways $C_{k-1, m+1}(n)$ may arise biologically is through the stochasticity of the neural response, where neurons fail to fire.  This stochasticity arises naturally from our system, the nucleus laminaris (NL) of the barn owl, in at least two different ways:  

\begin{enumerate}
\item{\emph{Stochasticity in the stimulus.} Recall that in our description of the NL (see Figure~\ref{AuditorySystem}b), we assumed that a neuron fires if the difference in the phase of the signal coming to the neuron is less than $\varepsilon$.  We also showed that within a column, neurons differ by a time difference of $\delta_l + \delta_r$.  This means that for a given $\varepsilon$, up to $\lceil \frac{\varepsilon}{\delta_l + \delta_r} \rceil$ neurons could fire.  Depending on the incoming time difference, one fewer neuron than this upper bound could fire.  Thus, the stochasticity of the signal in time could give rise to the additional codewords necessary for the convex closure.}
\item{\emph{Stochasticity in neural firing.}  In addition to the stochasticity of the incoming signal, there is stochasticity in the neural response; a neuron may fail to fire when it should (false negative), or fire when it should not (false positive), based on the stimulus.}
\end{enumerate}

To attain the convex closure and form the words in $C_{k-1, m+1}(n)$, we clearly need some neuron to fail to fire, but adding codewords resulting from other neurons failing to fire could potentially alter the convexity of the code.  Lemma~\ref{subcodePreserve} guarantees that if the convex closure has been attained, adding codewords that preserve the simplicial complex, or equivalently result from neurons failing to fire, maintains the convexity of the convex closure.

\begin{lemma}
\cite[Theorem~1.3]{NewPaper}
Let $C$ be a convex code.  If $\Delta(C) \supseteq \widetilde{C} \supseteq C$, then $\widetilde{C}$ is convex.
\label{subcodePreserve}
\end{lemma}

Thus, we can attain the convex closure of a periodic code by introducing some probability that neurons fail to fire; let $p$ be the probability that a neuron fires correctly so $1-p$ is the probability that the neuron fails to fire.  Given this probability and the fact that other neurons' failure to fire does not create non-convexity, we ask how likely it is to attain the convex closure via failure to fire stochasticity.

\begin{proposition}
\label{totalProb}
{Suppose each codeword in $C_{k,m}(k+m)$ is sampled $N$ times.  The probability of receiving all of the codewords in $C_{k-1,m+1}(k+m)$ is $$\mathcal{P}(k, m, N, p) = \left(1-(1-\frac{2}{k+m}(1-p)p^{k-1})^N\right)^{k+m}.$$}
\end{proposition}

Before proving this formula in the following section, we use it to estimate the number of times $N$ that each codeword needs to be sampled to attain a probability greater than .95 of seeing all the codewords in $C_{k-1, m+1}(k+m)$, rendering the resulting code convex. The results are plotted in Figure~\ref{kbehavior}, showing that $N$ grows supralinearly with $k$ but approximately linearly with $m$.  In exploring the convex closure, we assumed that $k\leq m$ to reflect the fact that neural codes are generally sparse.  The relationship between sparsity and convexity has only begun to be investigated \cite{YoungsInvolve}.  By comparing the way $N$ grows in comparison to $k$ and $m$, we argue that this sparsity is not only a general property of codes but is necessary for convexity to arise through stochasticity in this system.  By increasing $k$, the number of consecutive active neurons, the number of times each codeword must be sampled increases rapidly.  In contrast, a code can be expanded by increasing $m$, the number of consecutive silent neurons, without the number of samples of each codeword growing too rapidly that all the codewords in the convex closure would likely never all be seen.  Thus, sparsity allows the convex closure to be obtained through failure to fire stochasticity.  The importance of sparsity to the ability to achieve the convex closure is further seen in Figure~\ref{kbehavior}c.  Here, we see that there is an optimal value for $1-p$, the rate of failure to fire, which minimizes the number of times each codeword needs to be sampled to achieve the convex closure.  Intuitively, this optimal $1-p$ results from the fact that there must be some failure rate so that single neurons misfire, but if the failure rate is too high, multiple neurons will misfire at the same time.  By finding the minimum of $N$ with respect to $p$, we are able to see how this relates to the sparsity of the code.

\begin{corollary}
\label{optimalP}
The optimal failure to fire rate $1-p$ for which both $N$ is minimized for a given $\mathcal{P}$ and $\mathcal{P}$ is maximized for a given $N$ is
$$1-p = \frac{1}{k}.$$
\end{corollary}

Observe, that this result is not surprising; the number of neurons which fail to fire is binomially distributed as $\mathrm{Bin}(k, 1-p)$, so for the expected number of neurons that fail to fire to be exactly one, we need $1-p = \frac{1}{k}$.  This provides an additional argument for sparsity in our code because as $k$ increases, the probability of failure to fire decreases, meaning that the biological system must be increasingly precise in its firing as the number of active neurons increases.

\begin{figure}
 \begin{center}
  \includegraphics[width=\textwidth]{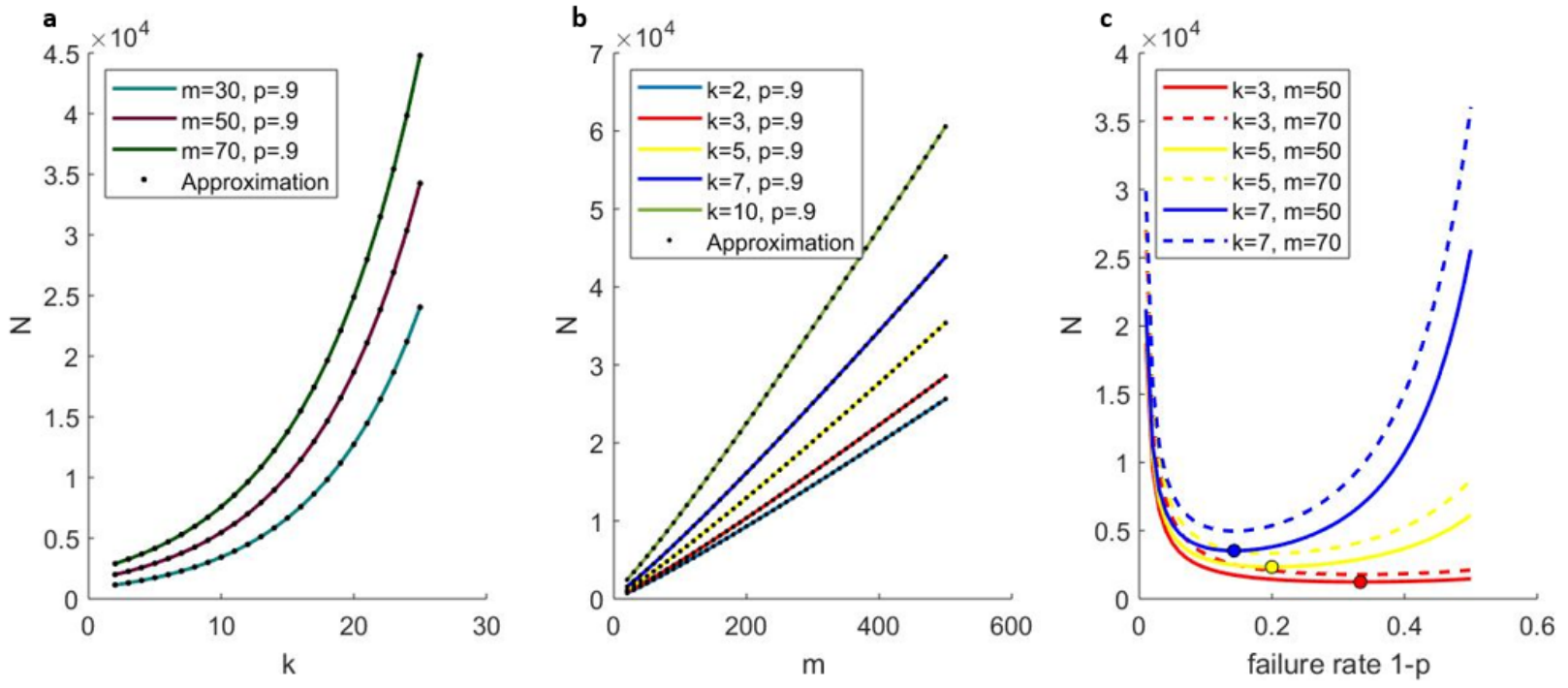}
  \end{center}
\caption{ \footnotesize \textit{Number of Samples N of Each Codeword Needed for $\mathcal{P}>.95$.} \textbf{(a)} \textit{Effect of k on N.}  For fixed $m$ and $p$, as $k$ increases, the number of times, $N$, that each codeword needs to be sampled for the probability that each codeword in the convex closure has been received grows supralinearly.  This suggests the importance of sparsity in keeping $N$ small. The dotted black lines show our approximation, which closely tracks the exact solutions (colored lines). \textbf{(b)} \textit{Effect of m on N.}  For fixed $k$ and $p$, as $m$ increases, $N$ grows approximately linearly.  Compared with the rapid growth of $N$ with $k$, expanding the code with additional silent neurons does not dramatically increase the number of times that each codeword must be sampled for a convex realization.  As before, the dotted black lines are our approximation.  \textbf{(c)} \textit{Effect of Failure to Fire Rate $1-p$ on N.} For fixed $k$ and $m$, we see that there is an optimal failure rate $1- p$ for which $N$ is minimized, $1-p = \frac{1}{k}$ (solid dots).}
\label{kbehavior}
\end{figure}

\subsection{Proof of Proposition~\ref{totalProb}}
We assume the probability of a neuron firing correctly is $p$, so the probability of a 1 being switched to a 0 is $1-p$.  Here, we consider only the case where neurons fail to fire; neurons have 0 probability of \textit{misfiring} (i.e., firing when they should be silent).  In addition, we assume that the stimulus space has a uniform distribution over all possible stimuli, so all codewords in $C_{k,m}(n)$ are equally probable to be the correct codeword, and we assume that the brain always stores the codewords from $C_{k,m}(n)$ in memory.  Let $N$ be the number of times that each codeword is sampled that the brain stores.

In Theorem~\ref{subcompletionTheorem}, we showed that for $1< k \leq m$, the convex closure is $C_{k,m}(n) \cup C_{k-1,m+1}(n)$.  We give the probability of this convex closure for $n=k+m$.  We first consider the probability of seeing one of the needed mandatory codewords.
 
\begin{lemma}
\label{singleProb}
{Let $c \in C_{k-1,m+1}(k+m)$.  Assuming that neurons never misfire and that the stimulus space is uniform such that all codewords in $C_{k,m}(k+m)$ are equally probable to be the true sent codeword, the probability that $c$ is received on a given trial is $$\frac{2}{k+m}(1-p)p^{k-1}.$$}
\end{lemma}

\begin{proof}
Without loss of generality, let $c = 1_1 \cdots 1_{k-1}0_k\cdots 0_{k+m}$.  Observe that, since we assume that neurons never misfire, $c$ can only be formed by the failure of one of the neurons in a codeword in $C_{k,m}(k+m)$ failing to fire.  There are exactly two codewords in $C_{k,m}(k+m)$ where the failure of one neuron produces $c$, the codewords $c_1 = 1_1\cdots1_k0_{k+1}\cdots0_{k+m}$ and $c_2 = 1_1\cdots1_{k-1}0_k\cdots0_{k+m-1}1_{k+m}$.  We have $\Pr(\text{$c$ received}\mid\text{$c_1$ or $c_2$ sent}) = \frac{\Pr(\text{$c$ received})}{\Pr(\text{$c_1$ or $c_2$ sent})}$.  Since we assumed the stimulus space was uniform, each codeword is equally likely to have been sent, so $\Pr(\text{$c_1$ or $c_2$ sent}) = \frac{2}{k+m}$.  If $c_1$ or $c_2$ was sent and $c$ was received, then exactly one neuron failed to fire and all the other neurons fired correctly, so $\Pr(\text{$c$ received}\mid\text{$c_1$ or $c_2$ sent})=(1-p)p^{k-1}$.  Therefore, we have $\Pr(\text{$c$ received})= \frac{2}{k+m}(1-p)p^{k-1}$.
\end{proof}

We now consider the set of codewords in $C_{k-1,m+1}(k+m)$.

\begin{lemma}
\label{codeProb}
{Let $q$ be the probability of seeing a codeword in $C_{k-1, m+1}(k+m)$ on a given trial.  If each of the codewords in $C_{k,m}(k+m)$ is sampled $N$ times, the probability of seeing all the codewords in $C_{k-1, m+1}(k+m)$ is $(1-(1-q)^N)^{k+m}$.}
\end{lemma}

\begin{proof}
Let $x$ be the probability of seeing a codeword at least once.  There are $k+m$ codewords in $C_{k,m}(k+m)$, so the probability of seeing all $k+m$ codewords at least once is $x^{k+m}$.  We have 
$$x = 1 - \Pr(\text{never seeing a codeword in $N$ trials}).$$
We also have 
$$\Pr(\text{never seeing a codeword in $N$ trials}) = (1-q)^N.$$
Thus, the probability of seeing all the codewords is 
$$x^{k+m} = \left(1 - (1-q)^N\right)^{k+m}.$$
\end{proof}

We are able to combine the results of Lemma~\ref{singleProb} and Lemma~\ref{codeProb} to give us the probability of receiving the convex closure,
$$\mathcal{P}(k, m, N, p) = \left(1-\left(1-\frac{2}{k+m}(1-p)p^{k-1}\right)^N\right)^{k+m}.$$

From this formula for the probability, it is natural to ask how many times each codeword needs to be sampled to achieve some probability that all the codewords in the convex closure have been received.  Using this formula, we derive the number of times $N$ which each codeword needs to be sampled in order to achieve some probability bound $\mathcal{P}$ of receiving all the codewords in the convex closure, finding
$$N = \log_{1-\frac{2}{k+m}(1-p)p^{k-1}}{\left(1-\mathcal{P}^{\frac{1}{k+m}}\right)} = \frac{\ln\left(1-\mathcal{P}^{\frac{1}{k+m}}\right)}{\ln\left(1-\frac{2}{k+m}(1-p)p^{k-1}\right)}.$$

To develop better intuition about the dependence of $N$ on $k$ and $m$, we can approximate $\ln\left(1-\frac{2}{k+m}(1-p)p^{k-1}\right)$ as $-\frac{2}{k+m}(1-p)p^{k-1}$.  Using this approximation, we find
$$N \approx \frac{\ln(\mathcal{P})+(k+m)\ln\left(\frac{1}{\mathcal{P}^{\frac{1}{k+m}}} - 1\right)}{-2(1-p)p^{k-1}} = f(k) + mg(k)\ln\left(\frac{1}{\mathcal{P}^{\frac{1}{k+m}}} - 1\right),$$
where $f(k)= \frac{\ln{\mathcal{P}}+k}{-2(1-p)p^{k-1}}$ and $g(k)= \frac{1}{-2(1-p)p^{k-1}}$.  The plots show that this approximation closely follows the analytic solution.  Through this approximation, we are able to see why we might expect $N$ to depend almost linearly on $m$, as seen in the plot in Figure~\ref{kbehavior}.

Recall our expression for $\mathcal{P}$ from Proposition~\ref{totalProb}.  By optimizing this expression for fixed $k$, $m$, and $N$, we are able to solve for the $p$ which gives the highest probability of observing the convex closure (Corollary~\ref{optimalP}).  The probabilty of observing the convex closure $\mathcal{P}$ has a minimum of 0 at $p=1$ and a maximum at $p = \frac{k-1}{k}$.  Similarly, by optimizing $N$ for fixed $k$, $m$, and $\mathcal{P}$, we find that $N$ is minimized at $p= \frac{k-1}{k}$.

These results are limited to the case where the stochasticity is only for neurons failing to fire with the assumption that neurons never fire when they should not. The more challenging question combinatorially is what the probability is that a code becomes convex when there is some nonzero probability that neurons fire when they should not because this additional firing changes the simplicial complex and hence which codewords are mandatory.  This combinatorial question remains open and is further complicated by the fact that for $n>k+m$, receiving a $(k-1)$-$(m+1)$ codeword requires a pattern of repeated errors at each firing band.

\subsection{Convex completions of Hamming distance $d$}
While less probable, it is also possible for neurons to fire incorrectly, which would correspond to codewords of greater weight that no longer preserve the simplicial complex.  Observe that while multiple errors in firing are possible, the probability of each additional error decreases.  We use Hamming distance as a measure of the degree to which two codewords differ, counting the number of errors that would be needed for one codeword to be transmitted as another.  Recall that the \textit{Hamming distance} between two codewords $a$ and $b$, denoted $d_H(a,b)$, is given by $d_H(a,b) = w_H(a-b)$ where subtraction is performed over $\mathbb{F}_2$.  Unlike in Theorem~\ref{subcompletionTheorem}, we no longer require that the new code preserve $\Delta(C)$, but instead require that the added codewords have small Hamming distance from the original codewords.
 
\begin{definition}
A \textit{Hamming distance d convex completion of $C$} is a code $\widehat{C}\supseteq C$ such that $\widehat{C}$ is convex, and for all $a\in \widehat{C}\setminus C$ there exists $c\in C$ such that $d_H(a,c) \leq d$.  We say $\widehat{C}$ is \textit{minimal} if $|\widehat{C}|$ is minimal.
\end{definition}

From \cite[Lemma~2.5]{CurtoConvex}, we know that any code which contains the all-ones codeword, $11\cdots1$, is convex.  Thus, for any code on $n$ neurons where the maximal weight codeword has weight $w$, we have a minimal Hamming distance $n-w$ convex completion given by simply adding the all-ones codeword to the code.  In particular, for the case of a $\km$ periodic code on $k+m$ neurons, we have a minimal Hamming distance $m$ convex completion given by adding the all-ones codeword. We can also guarantee a Hamming distance $k-1$ convex completion by adding all codewords which are subsets of some codeword in $C_{k,m}(k+m)$, but this method is rarely minimal.

As $d$ increases, the probability of a codeword of Hamming distance $d$ from an original codeword being received decreases, so while codewords of Hamming distance $k-1$ and $m$ are possible, these convex completions are often less probable.  For this reason, we give special attention to the cases of a Hamming distance 1 convex completion of a $\km$ periodic code.

Observe that for $k\leq m$, the convex closure of $C_{k,m}(k+m)$ is a Hamming distance 1 convex completion, where the codewords of $C_{k-1, m+1}(k+m)$ result from a single neuron in a codeword in $C_{k,m}(k+m)$ failing to fire (Theorem \ref{subcompletionTheorem}).  On the other hand, for $k \leq m-2$, the code obtained by adding higher weight codewords $C_{k,m}(k+m)\cup C_{k+1, m-1}(k+m)$ is a Hamming distance 1 convex completion, since this code is precisely the convex closure of $C_{k+1, m-1}(k+m)$.  Both of these examples of Hamming distance 1 convex completions require $k+m$ additional codewords, so it is natural to ask whether the convex closure is a minimal Hamming distance 1 convex completion.

\begin{theorem}{For $k\leq m$, the convex closure of $C_{k,m}(k+m)$ is a minimal Hamming distance 1 convex completion.}
\label{HammingSize}
\end{theorem}

\begin{proof}
Let $\widehat{C}$ be a Hamming distance 1 convex completion.  Define $A = \widehat{C}\setminus C$, so proving $\widehat{C}$ is minimal is equivalent to proving $|A|$ is minimal.

Let $C = C_{k,m}(k+m)$ for $k \leq m$ and let $\Delta = \Delta(C)$.  For the cases of $k=0$ and $k=1$, $\bar{C} = C$ (Theorem~\ref{subcompletionTheorem}), so $A$ is the empty set and must be minimal.  Now consider the case $1<k\leq m$.  We know that the convex closure is the minimal Hamming distance 1 convex completion which also preserves the simplicial complex of the code, where $A = C_{k-1, m+1}(k+m)$ and $|A| = k+m$.  Thus, in order to show that the convex closure is minimal, we will prove that there is no smaller $A'$ of Hamming distance 1 codewords such that $C_{\min}(\Delta(C \cup A')) \subseteq C \cup A'$.  Define $\widehat{\Delta} = \Delta(C \cup A')$ and $A^* = C_{\min}(\widehat{\Delta}) \setminus C$.  For $C\cup A'$ to be convex, we must have $A^* \subset A'$.

For any $A'$ such that $\widehat{\Delta} = \Delta$, we have $A^* = C_{k-1, m+1}(k+m)$.  We want to show that by adding a single Hamming distance 1 codeword to $A'$, we can reduce $|A^*|$ by at most 1. Without loss of generality, consider the face $\tau \in C_{k-1, m+1}(k+m)$ with $\tau = \sigma_{2, k}(k+m)$, using our notation from Section 3.2.  Recall from the proof of Lemma~\ref{subbandLink} that $\Lk_{\tau}(\Delta) = \{1,(k+1)\}$. To form a convex completion, we must either have $\tau \in A'$ or choose $A'$ such that $\Lk_{\tau}(\Delta(C \cup A'))$ is contractible.  Since $\Lk_\tau(\Delta) \subseteq \Lk_\tau(\widehat{\Delta})$, in order for  $\Lk_\tau(\widehat{\Delta})$ to be contractible, we must have that either the edge $\{1(k+1)\} \in \Lk_\tau(\widehat{\Delta})$ or the edges $\{1j\}$ and $\{(k+1)j\}$ are both in $\Lk_\tau(\widehat{\Delta})$ for some $j>k+1$.  The only way we can add the edge $\{1(k+1)\}$ to $\Lk_\tau(\widehat{\Delta})$ by adding a Hamming distance 1 codeword is if $\sigma_{1, k+1}(k+m) \in A'$.  The only way we can add the edges $\{1j\}$ and $\{(k+1)j\}$ using Hamming distance 1 codewords is if $1\cdots 1_{k}0\cdots01_j0\cdots0 \in A'$ and $01\cdots1_{k+1}0\cdots01_j0\cdots0\in A'$.  The addition of these codewords does not change the link of any of the other codewords in $C_{k-1, m+1}(k+m)$ and hence $A^* = (C_{k-1, m+1}(k+m) - \{\tau\}) \cup C_{\min}(\widehat{\Delta}))\setminus C$.  This gives us $|A^*| \geq |C_{k-1, m+1}(k+m)|-1$.  So for $A'$ to be a convex completion, we must have $A' \supseteq \sigma_{1, k+1}(k+m) \cup A^*$, giving us $|A'| \geq 1+|A^*| \geq k+m$.  Thus, there exists no smaller $A'$ such that $C_{\min}(\widehat{\Delta}) \subseteq C \cup A'$.   
\end{proof}

\section{Algebraic signatures of $\km$ periodic codes}
We defined $\km$ periodic codes as the codes containing all $\km$ periodic codewords, relying on a specific ordering of the vertices.  We showed that we could determine whether another maximal code is permutation equivalent to a periodic code by comparing the simplicial complexes of the codes (Proposition~\ref{simpPermEquivalent}).  In this section, we prove Theorem \ref{fullCF}, which gives an algebraic description of periodic codes and allows us to check if any code is permutation equivalent to a periodic code.

\subsection{The neural code as an algebraic ideal}
The code may also be viewed from an algebraic perspective as an ideal.  To encode a neural code $C$ as an ideal, we associate to each neuron an indeterminant $x_i$. The \textit{neural ideal} is defined by
$$J_C = \{f \in \mathbb{F}_2[x_1, \ldots, x_n] \mid f(c)=0 \text{ for all } c\in C\}\setminus \beta,$$
where $\mathbb{F}_2[x_1, \ldots, x_n]$ is the ring of polynomials with coefficients in $\mathbb{F}_2$, the finite field with 2 elements $\{0,1\}$ and $\beta = \{ x_i(1-x_i)\}_{i=1}^n$ the set of Boolean generators.  

The neural ideal gives us information about the relationships among the receptive fields of the neurons as explained in the following lemma.

\begin{lemma}
\cite[Lemma~4.2]{CurtoRing}
Let $C$ be a neural code and $\mathcal{U}$ a collection of open sets (not necessarily convex) such that $C = C(\mathcal{U})$.  Then for any $\sigma, \tau \subset [n]$ such that $\sigma \cap \tau = \emptyset$,
$$\prod_{i\in\sigma} x_i \prod_{j \in \tau}{(1-x_j)} \in J_{C} \Leftrightarrow \bigcap_{i\in \sigma} U_i \subseteq \bigcup_{j \in \tau} U_j.$$
\end{lemma}

For example, in Figure~\ref{codeEx}, we see $U_3 \subset U_2$, so $x_3(1-x_2) \in J_{C}$.  In the previous lemma, the generators of the neural ideal are given as polynomials of the form $x_\sigma \prod_{i \in \tau}{(1-x_i)}$, which we call \textit{psuedo-monomials} when $\sigma \cap \tau = \emptyset$.  

Viewing the generators of the neural ideal from the perspective of receptive fields, we are able to observe some special properties in $J_{C_{k,m}(n)}$ that result from the periodicity property (Lemma \ref{periodicityProp}).  In particular, since for $i\equiv j \mod{(k+m)}$ we have $x_i = x_j$, we know that $x_i$ and $x_j$ are interchangeable in the elements of the ideal of a $\km$ periodic code.  We define a map $T_{ij}$ between pseudo-monomials, where $T_{ij}(f)$ is $f$ with $x_i$ and $x_j$ interchanged.  For example, $T_{ij}(x_ix_\ell(1-x_j)) = x_jx_\ell(1-x_i)$ and $T_{ij}(x_i(1-x_\ell)) = x_j(1-x_\ell)$.

\begin{lemma}
\label{sameSet}
Let $C = C_{k,m}(n)$ be $\km$ periodic with $n>k+m$.  For any $i,j \in [n]$, $i\equiv j \mod{(k+m)}$ if and only if $x_i(1-x_j) \in J_C$.  Furthermore, if $i \equiv j\mod{(k+m)}$, then for every $f \in J_C$, we also have $T_{ij}(f) \in J_C$.

\end{lemma} 
\begin{proof}
($\Leftarrow$) Assume $x_i(1-x_j) \in J_C$, so $U_i \subset U_j$.  Suppose $i \not \equiv j \mod{(k+m)}$, and that $i \equiv \tilde{\imath}\mod{(k+m)}$ and $j \equiv \tilde{\jmath}\mod{(k+m)}$ for $\tilde{\imath}$ and $\tilde{\jmath}$ less than $k+m$.  We can choose a permutation of $s_{k,m}$ such that $\tilde{\imath} = 1$ and $\tilde{\jmath}\neq 1$, so there exists a codeword where $1 = i = \tilde{\imath}\neq \tilde{\jmath} = j$.  Thus, $U_i \not \subset U_j$, a contradiction, so we must have $i \equiv j \mod{(k+m)}$.  ($\Rightarrow$) Assume that $i \equiv j\mod{(k+m)}$.  By Lemma \ref{periodicityProp}, $c_i = c_j$ for all codewords in $C_{k,m}(n)$.  This implies that neuron $i$ and neuron $j$ fire over exactly the same set, so equivalently $U_i \subset U_j$ and $U_j \subset U_i$.  These receptive field relationships correspond to the generators $x_i(1-x_j)$ and $x_j(1-x_i)$.  Moreover, since $U_i$ and $U_j$ are the same set, $x_i$ and $x_j$ are interchangeable in the generators of the canonical form, as occurs under the operation $T_{ij}$. 
\end{proof}

From this result, we are able to define an equivalance relation on $[n]$ from the generators of $J_C$.

\begin{lemma}
Let $C = C_{k,m}(n)$ with $n\geq k+m$.  The relation $i \sim j$ if $x_i(1-x_j) \in J_C$ defines an equivalence relation on $[n]$.
\label{equivRel}
\end{lemma}

\begin{proof}
We trivially have $x_i(1-x_i) \in J_C$, so $\sim$ is reflexive.  By Lemma~\ref{sameSet}, if $x_i(1-x_j) \in J_C$, then $i \equiv j \mod{(k+m)}$, so we also have $j \equiv i \mod{(k+m)}$, giving us $x_j(1-x_i) \in \CF(C)$ and $\sim$ is symmetric.  Again applying Lemma~\ref{sameSet}, if $x_i(1-x_j) \in J_C$ and $x_j(1-x_\ell)\in J_C$, then $i \equiv j \mod{(k+m)}$ and $j \equiv \ell \mod{(k+m)}$, which implies $i \equiv \ell \mod{(k+m)}$, so we must have $x_i(1-x_\ell) \in J_C$, giving us the transitivity of $\sim$.  
\end{proof}  

Observe that this equivalence relation is not true for a general receptive field code.  In our example from Figure~\ref{codeEx}, $x_3(1-x_2) \in J_{C(\mathcal{U})}$ but $x_2(1-x_3)$ is not.

We also observe that we can find the neural ideal of the $\km$ periodic code for $m>k$ from the neural ideal of the $\km$ periodic code with $k\leq m$.

\begin{lemma}
If $x_{\sigma}\prod_{i \in \tau}(1-x_i)\in J_{C_{k,m}(n)}$, then $x_{\tau}\prod_{j \in \sigma}(1-x_j)\in J_{C_{m,k}(n)}$.
\label{flip}
\end{lemma}

\begin{proof}
To form the neural ideal, we take the set of functions that evaluate to zero on all codewords in the code.  Given a codeword $c \in C_{k,m}(n)$, there is a corresponding codeword $c'\in C_{m,k}(n)$ such that $c_i \neq c'_i$ for all $i$.  This implies that any function that evaluates to 1 on all codewords in $C_{k,m}(n)$ evaluates to 0 on some codeword in $C_{m,k}(n)$.  
\end{proof}

Observe that the $m$-$k$ periodic code can be formed from the $\km$ periodic code by flipping every bit in every codeword.  Lemma \ref{flip} shows that combinatorially the information represented by bits which are 1's and bits which are 0's has  a certain equivalence.  Yet, this information is not equivalent topologically.  For example, for all $x>1$, $C_{1, x}(1+x)$ is convex but $C_{x, 1}(1+x)$ is not in general.  This implies that the information represented by 1's in a code is fundamentally different than that represented by 0's.

In order to compare different codes, it is convenient to use a convention to represent the ideal of a code. In their work, Curto et al. \cite{CurtoRing} develop an algorithm which allows the neural ideal to be expressed in \textit{canonical form}.

\begin{definition}
Let $C$ be a neural code and $J_C$ its neural ideal.  The \textit{canonical form} of the neural ideal is the set of all minimal pseudo-monomial elements in $J_C$, where an element $f \in J_C$ is \textit{minimal} if $f \neq gh$ for any pseudo-monomial $g \in J_C$ with $\deg(g) < \deg(f)$ and some $h \in \mathbb{F}_2[x_1, \ldots, x_n]$.   
\end{definition}

From this canonical form, a description of the receptive field structure can be extracted from knowledge only of the code \cite{CurtoRing}.  The canonical form of the code in Figure~\ref{codeEx} is given by $\{x_3(1-x_2)\}$, corresponding to the receptive field relationship $U_3 \subset U_2$.

\subsection{The canonical form of $\km$ periodic codes}
One question of algebraic interest is whether the canonical forms of $\km$ periodic codes have any significant properties.  In particular, can the canonical form be used to detect whether a code is periodic?  In this section, we first give the canonical form for a periodic code on $k+m$ neurons.  We then prove several lemmas which extend these results from codes on $k+m$ neurons to codes on $n$ neurons, allowing us to present the canonical form of any $\km$ periodic code in Theorem~\ref{fullCF}.

We first introduce the definition of the \textit{interval mod n} between two indices, which will be useful in our formulation of the canonical form of a periodic code.  Recall our notation, 

$$\sigma_{i,j}(n) = \begin{cases}
\{\ell \mid i \leq \ell \leq j\} & \text{if } i\leq j\\
\{\ell \mid i \leq \ell \leq n\} \cup \{\ell \mid 1\leq \ell \leq j\} & \text{if } i>j
\end{cases}
$$

\begin{definition}
The \textit{interval mod n} between indices $i$ and $j$, denoted $\Int[i,j]$, is the set $\sigma_{i,j}(n)$ or $\sigma_{j,i}(n)$, whichever is smaller.  If $|\sigma_{i,j}(n)| = |\sigma_{j,i}(n)|$, we choose $\Int[i,j] = \sigma_{i,j}(n)$ such that $i<j$ by convention.  
\end{definition}

For example, we have $\text{Int}_5[1,2] = \{1, 2\}$, and $\text{Int}_5[1,5] = \{1, 5\}$.

For simplicity, we include the zeros codeword when we give the general structure of the canonical form.  The addition of the zeros codeword removes generators of the form $\prod_{i\in\tau}(1-x_i)$ with no changes to any of the other generators.  This observation in combination with our interval notation allows us to define four natural classes of pseudo-monomial generators of the canonical form of the periodic code $C_{k,m}(n)$.  We define

$$A_1 = \{x_ix_j \mid k<|\Intkm[i,j]|\}.$$

The set $A_1$ consists of generators of the form $x_ix_j$.  A generator $x_ix_j$ corresponds to $U_i \cap U_j = \emptyset$.  We know that the receptive field of two neurons intersect if and only if they both fire in the same codeword.  In $C_{k,m}(k+m)$, neurons $i$ and $j$ only cofire if $|\Intkm[i,j]|\leq k$, so for $U_i \cap U_j = \emptyset$, $|\Intkm[i,j]|>k$.  Thus, $A_1$ consists of generators of the neural ideal. We define

$$A_2 = \{x_ix_j(1-x_z) \mid z\in \Intkm[i,j] \text{ and } k \geq |\Intkm[i,j]|\}$$
and
$$A_3 = \{x_z(1-x_i)(1-x_j) \mid z\in \Intkm[i,j] \text{ and } k \geq |\Intkm[i,j]|\}.$$

We observe that $A_2$ also consists of generators of the neural ideal.  If both neuron $i$ and $j$ are firing and $|\Intkm[i,j]| \leq k$, then any neuron contained in $\Intkm[i,j]$ must also fire or there would be a band of firing neurons of size less than $k$.  Analogously, $A_3$ also consists of generators of the neural ideal since if both neuron $i$ and $j$ are not firing, then any neuron contained in $\Intkm[i,j]$ must also not fire or there would be a band of firing neurons of size less than $k$. We define

\begin{align*}
A_4 = &\{x_ix_jx_z \mid z \notin \Intkm[i,j], j\notin \Intkm[i,z], i\notin \Intkm[j,z],\\
&\text{ and } k \geq \max (|\Intkm[i,j]|, |\Intkm[i,z]|, |\Intkm[j,z]|)\}.
\end{align*}

$A_4$ also consists of generators of the neural ideal.  A generator $x_ix_jx_z$ corresponds to $U_i \cap U_j \cap U_z = \emptyset$.  For this generator to be minimal, we must have that the pairwise intersections are nontrivial, so $|\Intkm[i,j]|\leq k$, $|\Intkm[i,z]|\leq k$, and $|\Intkm[j,z]|\leq k$, but for the triple intersection to be trivial, we must have that the third vertex is not contained in these pairwise intervals.  Using these sets, we can construct the canonical form of the $\km$ periodic code on $k+m$ neurons. 

\begin{proposition}

Let $k\leq m$.  The canonical form of a $\km$ periodic code on $k+m$ neurons is given by
$$\CF(C_{k,m}(k+m) \cup \{\mathbf{0}\}) = A_1 \cup A_2 \cup A_3 \cup A_4.$$

\label{canonicalForm}
\end{proposition}

\begin{proof}
From the discussion above, we have seen that all of the described sets must be generators of the neural ideal.  It remains to show that this set is minimal and that there are no other generators.  It is clear that none of the generators in $A_2$ or $A_4$ are multiples of the generators of $A_1$ since in $A_1$, we have $|\Intkm[i,j]|>k$ and in $A_2$ and every pair in $A_4$, the interval has size less than or equal to $k$.  Thus the generators in $A_1$, $A_2$, and $A_4$ are minimal.  To see that $A_3$ is minimal, we note that there can be no generators of the form $x_i(1-x_j)$ corresponding to $U_i \subset U_j$ since the cyclic nature of the code makes it so that no neuron always fires when another is firing.  We show that this set is complete by showing that it generates no codewords not in $C_{k,m}(k+m) \cup \{\mathbf{0}\}$.  The all-zeros codeword clearly satisfies the conditions of the minimal generators and is included in $C$, so any other codeword must have a 1 at some bit.  Let $c_1 \cdots c_n$ be a binary string which vanishes on all of the generators.  Without loss of generality, let $c_1 = 1$.  We can choose $c_2$ to be 0 or 1.  If we choose it to be zero, then we must choose $c_{n-k+2}= \cdots =c_n = 1$ to vanish on the generators in $A_2$ and $A_3$.  We also must have that all other bits are 0 to vanish on the generators in $A_1$.  Thus, we generate a $\km$ periodic codeword.  If we choose $c_2 = 1$, then we can choose $c_3$ to be 0 or 1, and if we choose it to be 0, we introduce analogous restrictions on the remaining bits in the codewords as when we chose $c_2 = 0$, so we form another $\km$ periodic codeword but shifted by one bit.  Thus, whenever we choose $c_j = 0$ for $j<k$, given that $c_1=1$, we have fixed the remaining bits of the code so that we have a $\km$ periodic codeword.  Thus, we do not generate any codewords other than those in $C$.  Therefore, the set $A_1 \cup A_2 \cup A_3 \cup A_4$ is complete and consists of minimal generators of the neural ideal of $C_{k,m}(k+m)$.
\end{proof}

Given that the convex closure is closely related to the original code by the union with another periodic code, it is natural to ask if we can also find the canonical form of the convex closure.  To do this we will use Algorithm 1 of \cite{YoungsAlgo}, which describes a method to update the canonical form of a code, $\CF(C)$ when a new codeword $c$ is added.  For $f \in \CF(C)$, if $f(c) = 0$, add $f$ to a set $L$, and otherwise, add $f$ to a set $M$.  For every $g \in M$ and every $c_i$, define $h = (x_i-c_i)g$.  If $h$ is not a multiple of an element of $L$ and $g$ is not a multiple of $(x_i-c_i-1)$, add $h$ to a set $W$.  The canonical form of the new code is given by $\CF(C \cup c) = L \cup W$.

We also require defining a subset of the generators in $A_3$ which are not generators of the neural ideal of the closure,
$$\tilde{A_3} = \{x_z(1-x_i)(1-x_j) \mid z\in \Intkm[i,j] \text{ and } k = |\Intkm[i,j]|\}.$$

\begin{lemma}
Let $C = C_{k,m}(k+m)$ for $k\leq m$.  The canonical form of the convex closure $\bar{C} \cup \{\mathbf{0}\}$ is
$$\CF(\bar{C} \cup \{\mathbf{0}\}) = \CF(C \cup \{\mathbf{0}\})\setminus \tilde{A_3}.$$
\label{closureCF}
\end{lemma}

\begin{proof}
We have $\bar{C} \cup \{\mathbf{0}\} = C_{k,m}(k+m) \cup C_{k-1, m+1}(k+m) \cup \{\mathbf{0}\}$.  Algorithm 1 of \cite{YoungsAlgo} allows us to determine the canonical form of a code that is modified by adding a single codeword.  Let $c \in C_{k-1, m+1}(k+m)$.  Since $c$ vanishes on every generator in $A_1$, $A_2$, and $A_4$, we have $L = A_1 \cup A_2 \cup A_4$.  Since $C_{k-1, m+1}(k+m)$ is periodic, we know that $c$ also vanishes on every generator in $A_3\setminus\tilde{A_3}$.  It remains to show that for every $f \in \tilde{A_3}$, there exists $c \in C_{k-1, m+1}(k+m)$ which does not vanish on $f(c)$.  Note that since $|\Intkm[i,j]| = k$, we have $j = i+k-1$.  Take the codeword $\sigma = \sigma_{i, i+k-2}(k+m)$.  We have $f(\sigma) = 1$.  So we add $x_z(1-x_i)(1-x_j)$ to $M$.  We have $c_\ell = 1$ for $\ell \in \Int[i, i+k-2]$, but $x_z(1-x_i)(1-x_j)(1-x_\ell)$ is a multiple of a generator in $A_3\setminus\tilde{A_3}\subset L$.  We have $c_\ell = 0$ for $\ell \not \in \Intkm[i, i+k-2]$.  For $\ell = i+k-1$, $x_zx_\ell(1-x_i)(1-x_j)$ is a multiple of a generator in $A_2 \subset L$, and otherwise we have $x_zx_\ell(1-x_i)(1-x_j)$ is a multiple of a generator in $A_1 \subset L$.  Thus, we have $\CF(\bar{C} \cup \{\mathbf{0}\}) = L = (A_1\cup A_2 \cup A_3 \cup A_4)\setminus \tilde{A_3} = \CF(C \cup \{\mathbf{0}\})\setminus \tilde{A_3}$.
\end{proof}

Thus, we have the canonical form of a $\km$ periodic code and its closure on $k+m$ neurons, and we want to extend this to a $\km$ periodic code on $n$ neurons.  In particular, Lemma~\ref{sameSet} gives us the key result that allows us to do so.  Observe that by Lemma~\ref{periodicityProp}, the convex closure also satisfies that $c_i = c_j$ for $i \equiv j \mod{(k+m)}$, so the same lemma allows us to extend the results of Lemma~\ref{sameSet} to the canonical form of the convex closure.  Lemma~\ref{sameSet} also allows us to define an equivalence relation on the generators of the neural ideal of a $\km$ periodic code.

More significantly, Lemma~\ref{sameSet} in combination with Proposition~\ref{canonicalForm} allows us to detect if a code of arbitrary length is periodic, as we will show in the following section.  To do so we will introduce a concept of equivalence of pseudo-monomials.  

\begin{definition}
Let $f$ and $g$ be pseudo-monomials, $f = x_\sigma \prod_{i\in \tau} (1-x_i)$ and $g = x_{\sigma'} \prod_{i\in \tau'}(1-x_i)$.  We say $f\equiv g\mod{a}$ if there exist bijections $s: \sigma \rightarrow \sigma'$ and $t: \tau \rightarrow \tau'$ such that $s(i) \equiv i\mod{a}$ and $t(j) \equiv j\mod{a}$ for some integer $a$.
\end{definition}

Observe that $g\equiv g \mod{(k+m)}$ trivially by taking both $s$ and $t$ as the identity.

We note that the canonical form of the neural ideal fully characterizes the code \cite{CurtoRing}.  If two codes are permutation equivalent, we can similarly permute the indeterminants, $x_i$, that appear in the canonical form, so if one canonical form, $\CF(C)$ can be attained from the other, $\CF(C')$ through a permutation of the indeterminants, we say that the canonical forms are permutation equivalent, denoting this equivalence as $\CF(C) \cong \CF(C')$.  Combining the results of Proposition~\ref{canonicalForm}, Lemma~\ref{closureCF}, and Lemma~\ref{sameSet} allows us to give the canonical form of a $\km$ periodic code on $n$ neurons and that of its convex closure, hence allowing us to determine if any code is permutation equivalent to a periodic code.  We see that $\CF(C_{k,m}(n)\cup \{\mathbf{0}\}) \supset \CF(C_{k,m}(k+m)\cup \{\mathbf{0}\})$.  The canonical form $\CF(C_{k,m}(n)\cup \{\mathbf{0}\})$ also contains the generators which define the equivalence relation, and as a result of this equivalence, contains generators equivalent modulo $k+m$ to the generators of $\CF(C_{k,m}(k+m)\cup \{\mathbf{0}\})$. 

\begin{theorem}
Let $C_{k,m}(n)$ be a $\km$ periodic code on $n$ neurons and $\overline{C_{k,m}(n)}$ be its convex closure.  Define $B = \{x_i(1-x_j) \mid i \equiv j \mod{(k+m)} \text{ and } i\neq j\}$.  Then
\begin{align*}
\CF(&C_{k,m}(n) \cup \{\mathbf{0}\}) =\\
&\{f \mid f\equiv g\mod{(k+m)} \text{ for some } g\in \CF(C_{k,m}(k+m) \cup \{\mathbf{0}\})\}\cup B
\end{align*}
and
\begin{align*}
\CF(&\overline{C_{k,m}(n)} \cup \{\mathbf{0}\}) = \\
&\{f \mid f\equiv g\mod{(k+m)} \text{ for some } g\in \CF(\overline{C_{k,m}(k+m)} \cup \{\mathbf{0}\})\}\cup B.
\end{align*}
Moreover, a code $C$ of length $n$ which does not contain the all-zeros codeword is permutation equivalent to $C_{k,m}(n)$ if and only if 
$$\CF(C \cup \{\mathbf{0}\}) \cong \CF(C_{k,m}(n)\cup \{\mathbf{0}\})$$
for some permutation of the neurons in $C$.  Similarly, $C$ is permutation equivalent to the convex closure if and only if
$$\CF(C \cup \{\mathbf{0}\}) \cong \CF(\overline{C_{k,m}(n)}\cup \{\mathbf{0}\}).$$

\label{fullCF}
\end{theorem}

The canonical forms follow immediately from Proposition~\ref{canonicalForm}, Lemma~\ref{closureCF}, and Lemma~\ref{sameSet}. From the results in \cite{CurtoRing}, we know that $\CF(C)$ fully determines $J_C$, which, in turn, fully determines $C$.  Thus, $\CF(C \cup \{\textbf{0}\}) \cong \CF(C_{k,m}(n)\cup \{\textbf{0}\})$ if and only if $C$ and $C_{k,m}(n)$ are permutation equivalent.  Thus, we are able to provide three simple checks to detect that a code is not periodic.

\begin{lemma} 
Let $C$ be a neural code with canonical form $\CF(C \cup \{\mathbf{0}\})$.  If any of the following conditions hold, $C$ is not periodic:  

1.  There exists $f \in \CF(C \cup \{\mathbf{0}\})$ where $f = \prod_{i\in\sigma}x_i \prod_{j\in\tau}(1-x_j)$ such that $|\tau| > 2$ or $|\sigma \cup \tau| > 3$.

2.  For $i \neq j$ and $x_i(1-x_j) \in \CF(C \cup \{\mathbf{0}\})$, there exists $g \in \CF(C \cup \{\mathbf{0}\})$ such that $T_{ij}(g) \notin \CF(C \cup \{\mathbf{0}\})$.

3.  For $i \neq j$ and $j \neq \ell$, both $x_i(1-x_j) \in \CF(C \cup \{\mathbf{0}\})$ and $x_j(1-x_\ell) \in \CF(C \cup \{\mathbf{0}\})$, but $x_i(1-x_\ell) \notin \CF(C \cup \{\mathbf{0}\})$.

\label{easyCheck}
\end{lemma}

The proof of Lemma~\ref{easyCheck} follows immediately from Theorem~\ref{canonicalForm} and Lemma~\ref{equivRel}.  Observe that this guarantees that if $x_i(1-x_j)$ is in the canonical form, $x_j(1-x_i)$ is also in the canonical form, as required by Lemma~\ref{sameSet}, since either pseudo-monomial is obtained from the other by applying $T_{ij}$.  This lemma provides a simple way to determine when a code is not periodic.  Next, we will show that for arbitrary codes of length $n$ which satisfy these simple criteria, there is still a method which will allow us to determine $k$ and $m$ and hence the permutation equivalence of the code to a periodic code.

\subsection{Identifying periodic codes algebraically}
Our formulation of the canonical form of a $\km$ periodic code requires knowledge of $k$ and $m$.  For a given code, it may not be immediately obvious whether it is periodic as the neurons may have been permuted as we saw in Figure \ref{CyclicPerm}, where 

\begin{align*}
C' = \{10010000, &01001000, 00100100, 00010010, 00001001,\\ 
&10000100, 01000010, 00100001\}
\end{align*} and 
\begin{align*}
\tilde{C} = \{10100000, &01010000, 00101000, 00010100, 00001010,\\
 &00000101, 10000010, 01000001\},
\end{align*} 
which have corresponding canonical forms,
{\small
\begin{align*}
\CF(C' \cup \{\mathbf{0}\}) = \{&x_1x_2,          
x_2x_4,          
x_1x_5,          
x_4x_5,          
x_1x_3,          
x_2x_3,          
x_3x_4,          
x_3x_5,          
x_2x_6,          
x_4x_6,          
x_5x_6,          
x_1x_7,\\          
&x_3x_7,          
x_5x_7,          
x_6x_7,          
x_4(1-x_1)(1-x_7),
x_1x_8,          
x_2x_8,          
x_4x_8,          
x_6x_8,          
x_7x_8,\\          
&x_5(1-x_2)(1-x_8),
x_1(1-x_4)(1-x_6),
x_6(1-x_1)(1-x_3),\\
&x_2(1-x_5)(1-x_7),
x_7(1-x_2)(1-x_4),
x_3(1-x_6)(1-x_8),\\
&x_8(1-x_3)(1-x_5) \}
\end{align*}
}
and
{\small
\begin{align*}
\CF(\tilde{C} \cup \{\mathbf{0}\}) = \{&x_1x_2,          
x_2x_3,          
x_1x_4,          
x_3x_4,          
x_1x_5,          
x_2x_5,          
x_4x_5,          
x_3(1-x_1)(1-x_5),
x_1x_6,          
x_2x_6,\\          
&x_3x_6,          
x_5x_6,          
x_4(1-x_2)(1-x_6),
x_2x_7,          
x_3x_7,          
x_4x_7,          
x_6x_7,          
x_5(1-x_3)(1-x_7),\\
&x_1x_8,          
x_3x_8,          
x_4x_8,          
x_5x_8,         
x_7x_8,          
x_6(1-x_4)(1-x_8),
x_1(1-x_3)(1-x_7),\\
&x_7(1-x_1)(1-x_5),
x_2(1-x_4)(1-x_8),
x_8(1-x_2)(1-x_6)\}.
\end{align*}
}

In this case, it is not immediately obvious even from the canonical form whether the code is periodic as all of the generators of $\CF(C')$ and $\CF(\tilde{C})$ have the form $x_ix_j$ or $x_i(1-x_j)(1-x_z)$, consistent with Lemma~\ref{easyCheck}, but only $C'$ is periodic.  The information we do gain from the canonical form is that there are no generators of the form $x_i(1-x_j)$, so we know that if the code is periodic, then $n = k+m$ (Lemma \ref{sameSet}), which also allows us to determine $k$ by taking the weight of each codeword, but the question remains how to determine $k$ and $m$ in a general case.  

For a code which satisfies the easy-to-check conditions in Lemma~\ref{easyCheck}, we present an algorithm that allows us to determine $k$ and $m$, and hence whether an arbitrary code $C$ of length $n$ is $\km$ periodic for some permutation of the neurons.

\begin{enumerate}
\item{\textbf{Determine $k+m$ by determining equivalence classes of vertices.}  In Lemma~\ref{equivRel}, we showed that the relation $i\sim j$ if $x_i(1-x_j) \in J_C$ is an equivalence relation, so we can partition $[n]$ into $k+m$ equivalence classes.  If $|C| \neq k+m$, then $C$ is not periodic (Lemma~\ref{size}).}
\item{\textbf{Determine $k$ by forming $C\vert_{k+m}$.}  Let $C\vert_{k+m}$ be the code formed by restricting $C$ to $k+m$ vertices, where there is one vertex from each equivalence class. If $C\vert_{k+m}$ is not a constant weight code, then $C$ is not periodic.  Otherwise, $k = w_H(c)$ for $c \in C\vert_{k+m}$.}
\item{\textbf{Check for permutation equivalence given $k$ and $m$.}  Given $k$ and $m$, apply Theorem~\ref{fullCF} to determine if $C$ is permutation equivalent to $C_{k,m}(n)$.}
\end{enumerate}

We apply this algorithm to $C'$ and $\tilde{C}$ to show that $C'$ is periodic and $\tilde{C}$ is not.  In the first step of the algorithm, we find that each neuron is its own equivlance class for both codes, giving us $k+m = 8$ and both codes contain 8 codewords.  Since each neuron is its own equivalence class, we have that $C\vert_{k+m}$ is the original code for both cases in the second step of our algorithm.  Each codeword in both codes has weight 2, giving us $k=2$ for both codes.  In the third step, we check for permutation equivalence of the canonical forms of $C'$ and $\tilde{C}$ with the canonical form of $C_{2, 6}(8)$.  We have

{\small
\begin{align*}
\CF(C_{2,6}(8)\cup \{\mathbf{0}\}) = &\{ x_1x_3,          
x_1x_4,          
x_2x_4,          
x_1x_5,          
x_2x_5,          
x_3x_5,          
x_2(1-x_1)(1-x_3),\\
&x_4(1-x_3)(1-x_5),
x_1x_6,          
x_2x_6,          
x_3x_6,          
x_4x_6,          
x_3(1-x_2)(1-x_4),\\
&x_5(1-x_4)(1-x_6),
x_1x_7,          
x_2x_7,          
x_3x_7,          
x_4x_7,          
x_5x_7,          
x_6(1-x_5)(1-x_7),\\
&x_2x_8,          
x_3x_8,          
x_4x_8,          
x_5x_8,          
x_6x_8,          
x_7(1-x_6)(1-x_8),
x_1(1-x_2)(1-x_8),\\
&x_8(1-x_1)(1-x_7)\}.
\end{align*} 
}

We see by applying the permutation $(24)(37)(68)$ to $C'$ that we attain the same canonical form, so $C'$ is permutation equivalent to a periodic code.  There is no such permutation for $\tilde{C}$, so this code is not periodic.

\section{Discussion}
We showed that periodic codes, $C_{k,m}(n)$, with $1<k\leq m$ do not have a convex realization, potentially explaining the behavioral errors which owls make in localizing sounds of a single frequency.  However, for sounds with greater bandwidth, the owl is able to locate sounds with high precision, which suggests that there is a convex realization of these codes.  In particular, we showed that the convex closure of a single periodic code is its union with another periodic code. Such a code could arise by combining the code from many isofrequency columns as occurs in the inferior colliculus, perhaps explaining why the first space mapped cells exist in this nucleus.  Alternatively, we discussed that this could arise biologically through stochasticity, suggesting that both stochasticity and sparsity might be advantageous biologically.

Here we have framed our questions in terms of the system of sound localization in the owl, but we note that there are other systems which may be a natural extension of periodic codes.  For example, the receptive fields of rats' grid cells are centered at the vertices of a hexagonal lattice so are themselves periodic \cite{GridCells}.  This two dimensional system of grid cells raises the question of how to define periodic codes in higher dimensions, which we leave for further research.

\textbf{Acknowledgements}
We thank Michael Reed and William Pardon for helpful discussions in the early phases of this project.  We also thank Nell Cant for suggesting the barn owl's auditory system as a model system and Caitlin Lienkaemper for her comments on the manuscript.  This work was supported by NIH R01 EB022862 and NSF DMS1516881.

\bibliographystyle{spmpsci}      
\bibliography{Citations}   

%
%
\end{document}